%% file: main.tex
\documentclass[pdflatex,a4paper,ja=minimal]{bxjsarticle}

\usepackage{graphicx}
\usepackage{amsmath}
\usepackage{amssymb}
\usepackage{amsthm}
\usepackage{booktabs}
\usepackage{url}
\usepackage{here}
\usepackage{float}
\usepackage{placeins}
\usepackage{algorithm}
\usepackage{algpseudocode}
\usepackage{longtable}
\usepackage{pgfplots}
\usepackage{siunitx}
\usepackage{multirow}
\usepackage{xcolor}
\usepackage{colortbl}
\usepackage{subcaption}
\DeclareCaptionLabelFormat{spaced}{#1~#2}
\usepackage[hidelinks]{hyperref}
\pgfplotsset{compat=1.18}

\theoremstyle{plain}
\newtheorem{lemma}{Lemma}
\newtheorem{theorem}{Theorem}

\newcommand{\minVal}{\mathit{minVal}}
\newcommand{\maxVal}{\mathit{maxVal}}

\title{
Shared-Memory Range-Tiled CDF Sort\\
for Small-Range Integer Keys on GPUs
}
\author{
Kento Ando$^{1}$ \quad
Kaito Takase$^{2}$ \quad
Noriyuki Fujimoto$^{3}$ \quad
Koichi Wada$^{4}$\\[0.5em]
\small $^{1}$Major in Applied Informatics, Graduate School of Science and Engineering, Hosei University\\
\small $^{2}$THIRD, Inc.\\
\small $^{3}$Department of Core Informatics, Graduate School of Informatics, Osaka Metropolitan University\\
\small $^{4}$Information Media Education and Research Center, Hosei University
}
\date{}

\begin{document}
\maketitle

\begin{abstract}
We study unstable integer sorting on GPUs for arrays whose elements lie in a known integer range.
Focusing on counting-sort-based methods that determine the output interval of each value from its frequency and the prefix sums of the frequencies, we propose and evaluate Range-Tiled CDF sort (RT-CDF), which partitions the possible value range into small intervals, called tiles, that fit in shared memory.
For each tile, RT-CDF constructs a histogram, computes its prefix sum as a local CDF, and directly generates the output array from the local CDF.
We compare RT-CDF against three baselines: CUB \texttt{DeviceRadixSort}, whose processed bit range is restricted to $[0,\lceil\log_2 R\rceil)$ to exploit the known range size $R$; Ref-H-P sort; and an implementation based on the algorithm of Kolonias et al.
Experiments on an NVIDIA GeForce RTX 4090 with range sizes from $R=2^7$ to $2^{18}$, input sizes from $n=10^6$ to $10^9$, and uniformly distributed, normally distributed, and all-equal inputs show that RT-CDF outperforms the baselines over a broad set of conditions for small to medium ranges, achieving a maximum speedup of 4.39 over the fastest baseline.
For $R=2^{18}$, however, at least one baseline outperforms RT-CDF for every evaluated input size and input distribution, showing that the cost of histogram construction limits the applicability of RT-CDF to larger ranges.
\end{abstract}

\input{sections/01_introduction}
\input{sections/02_background}
\input{sections/03_method}
\input{sections/04_experimental_evaluation}
\input{sections/05_conclusion}
\bibliographystyle{unsrt}
\bibliography{references}
\end{document}

%% file: sections/01_introduction.tex
\section{Introduction}
\label{sec:introduction}

Sorting is a fundamental operation used in many computational tasks, including database processing, graph processing, simulation, and image processing.
GPUs can process large numbers of elements in parallel, and many parallel sorting algorithms, including radix sort, merge sort, and sample sort, have been studied and implemented~\cite{arkhipov2017survey,satish2009gpusort}.
GPU radix sorting for integers is particularly important in practice, and recent implementations such as Onesweep reduce the amount of data read from and written to global memory~\cite{adinets2022onesweep}.
CUB \texttt{DeviceRadixSort}, provided by NVIDIA, is also widely used as a highly optimized GPU sorting implementation~\cite{nvidia_cub_radixsort}.

When the range of possible input values is known in advance and is sufficiently small relative to the input size, however, this additional restriction can be exploited.
In particular, counting sort can count the occurrences of each value and use prefix sums of the counts to determine the interval occupied by each value in the output array.
If stability is not required, individual input elements need not be moved to their respective ranks.
A correct sorted output can instead be obtained by determining the contiguous interval occupied by each value and writing that value throughout the interval.
We denote the number of input elements by $n$ and the size of the known integer range $[\minVal,\maxVal)$ by $R=\maxVal-\minVal$.
The problem considered in this paper is to sort $n$ integers from this range into nondecreasing order without preserving the relative order of equal values.

Previous work on range-restricted integer sorting on GPUs includes the counting sort of Kolonias et al.~\cite{kolonias2011countsort}, a GPU implementation of H-P sort based on Eisenstat's algorithm~\cite{eisenstat2007sumcrcw,kozakai2021hp}, and Ref-H-P sort (hereafter RefHP), which refines H-P sort for GPU execution~\cite{takase2024hpcasia}.
These methods share the property that a histogram and its prefix sums provide information about the output position or output interval of each value.
We focus on this common structure, in which output intervals are determined from a histogram and prefix sums.
We then propose and evaluate Range-Tiled CDF sort (RT-CDF), which partitions the possible integer range into multiple small intervals, called range tiles, that fit in shared memory, and constructs a histogram and a cumulative distribution function (CDF) for each tile.

In this paper, the term CDF denotes the cumulative-count array obtained as the prefix sum of a histogram.
It is not a probability distribution normalized to the interval $[0,1]$; rather, each entry gives the number of input values no greater than the corresponding value.
The CDF constructed within a tile, which we call a local CDF, determines the output interval occupied by each value in that tile.
Consequently, the output can be generated without rescanning the input array.

The contributions of this work are as follows.

\begin{itemize}
    \item We develop RT-CDF, which partitions the value range into tiles that fit in shared memory and directly generates the output array from the local CDF of each tile.
    \item We clarify the algorithmic relationship between RT-CDF and a counting-sort implementation based on Kolonias et al., H-P sort, and RefHP, and identify how RT-CDF differs from existing methods that determine output intervals from a CDF.
    \item On an NVIDIA GeForce RTX 4090, we compare RT-CDF with CUB \texttt{DeviceRadixSort}, RefHP, and an implementation based on Kolonias et al., and identify the range sizes and input sizes for which RT-CDF is effective, as well as its limitations near $R=2^{18}$.
\end{itemize}

%% file: sections/02_background.tex
\section{Background and Motivation}
\label{sec:background}

\subsection{Problem Setting and Counting-Sort-Based Baselines}
\label{subsec:problem-baselines}

We consider unstable sorting of arrays whose possible integer values lie in a known range.
The input is an array $A[0..n-1]$ of $n$ integers, all of which belong to a known range $[\minVal,\maxVal)$.
Let $R=\maxVal-\minVal$.
For simplicity of notation, we subtract $\minVal$ from every element and continue to denote the resulting array by $A$, so that $0 \leq A[i] < R$.
The output is an array $B[0..n-1]$ containing the elements of $A$ in nondecreasing order.
If value $v\in[0,R)$ occurs $c_v$ times, then its positions in the output array form the contiguous interval
\[
  \left[\sum_{u=0}^{v-1} c_u,\ \sum_{u=0}^{v} c_u\right).
\]
Thus, it is unnecessary to compute a separate rank for every input element; it suffices to write each value into its corresponding output interval.

The most basic method for this problem setting is counting sort.
It first constructs a histogram $H[0..R-1]$ satisfying $H[v]=c_v$, and then computes its prefix sums to determine the output interval of each value.
For an array $X[0..\ell-1]$, its prefix-sum array $Y[0..\ell-1]$ is defined by
\[
Y[i] = \sum_{j=0}^{i} X[j].
\]
Although this procedure is straightforward sequentially, histogram updates on a GPU may concentrate additions at the same address, in which case contention among atomic operations can dominate performance.
Furthermore, a stable counting sort must determine the local rank of each input element among elements with the same value and combine that rank with a prefix sum to determine its output position.
By contrast, the unstable integer-array sorting problem considered here only requires the output interval of each value, allowing the output-generation step to be simplified.

As an early study of counting sort on GPUs, Sun and Ma presented a CUDA implementation of counting sort.
Using a 2~GHz AMD Athlon 64 X2 3800+ CPU and an NVIDIA GeForce 9600 GSO GPU, they reported an 8.08-fold speedup over the corresponding CPU implementation for five million integers~\cite{sun2009countsort}.
Faujdar and Ghrera compared sequential counting sort with CUDA-based parallel counting sort under six input conditions.
Their sequential implementation ran on a 2.60~GHz Intel Core i5-3230M, whereas the CUDA implementation ran on a system with a 2.93~GHz Intel Core i3-530 and an NVIDIA GeForce GTX 460; they reported a maximum speedup of 66 for inputs containing from 1,000 to 10,000,000 elements~\cite{faujdar2016countsort}.
These studies show both that counting sort can be effective on GPUs and that its performance can depend strongly on the input distribution and the value range.

Kolonias et al. designed an integer counting sort for CUDA GPUs whose final output stage avoids synchronization by writing different values to disjoint intervals.
They also considered constructing one partial histogram per thread block in shared memory when the range fits in shared memory~\cite{kolonias2011countsort}.
Our method extends this shared-memory histogram idea through range tiling and differs further in that it directly determines the value at each output position from a local CDF.

Building on the idea of shared-memory partial histograms, we partition the range into small intervals of width $T$, called range tiles, that fit in shared memory, thereby increasing the range sizes that can be handled.
For each range tile, we construct a CDF---the cumulative counts within that tile---and load the resulting local CDF into shared memory to determine the value corresponding to each output position.

The output-generation method of Kolonias et al. can be simplified for our problem setting as shown in Algorithm~\ref{alg:kolonias-countsort}.
First, a histogram of the input array is computed, and its prefix sums give the output starting position $P[v]$ of each value.
Then, for every value $v$, the interval $[P[v],P[v+1])$ is filled with $v$.
Because different values are written to disjoint output intervals, the final step requires no synchronization between values.

\begin{algorithm}[H]
\caption{Simplified structure of the counting sort of Kolonias et al.~\cite{kolonias2011countsort}}
\label{alg:kolonias-countsort}
\begin{algorithmic}[1]
\Require Input array $A[0..n-1]$ and range size $R$
\Ensure Sorted array $B[0..n-1]$
\State $H[0..R-1] \gets$ histogram of $A$
\State $P[0] \gets 0$
\State $P[1..R] \gets$ prefix sums of $H$
\ForAll{$v\in[0,R)$ \textbf{in parallel}}
  \For{$p \in [P[v],P[v+1])$}
    \State $B[p] \gets v$
  \EndFor
\EndFor
\State \Return $B$
\end{algorithmic}
\end{algorithm}

Algorithm~\ref{alg:kolonias-countsort} is directly related to our method in that it generates the output array as one interval per value.
If the writing of output intervals is parallelized over values, however, the amount of work assigned to each value becomes imbalanced when occurrence counts vary substantially.
Kolonias et al. explain that, under a uniform distribution, the average number of iterations is approximately $N/k$, whereas if 50\% of the input has the same value, the computation assigned to that value must perform $N/2$ writes~\cite{kolonias2011countsort}.
Our implementation retains the idea of output intervals for individual values, but assigns parallelism to output positions and determines the value of each output position from a local CDF loaded into shared memory.
The goal is to avoid the extreme workload imbalance that can arise when one processing unit is assigned to each value.

One practical baseline in this study is CUB \texttt{DeviceRadixSort} (hereafter CUB)~\cite{nvidia_cub_radixsort}.
CUB allows the bit range used for sorting to be specified by \texttt{begin\_bit} and \texttt{end\_bit}.
Because the normalized range size $R$ is known, we provide CUB with the relevant bit range and invoke it with \texttt{begin\_bit}$=0$ and \texttt{end\_bit}$=\lceil\log_2 R\rceil$.
This is a more appropriate baseline for our problem than always processing all 32 bits.
CUB is, however, a general-purpose implementation that provides stable sorting, a stronger functionality than the unstable integer sorting considered here.
The comparison with CUB therefore evaluates how effective an implementation specialized to our more restricted problem can be relative to a practical general-purpose GPU sorter.

\subsection{H-P Sort, RefHP, and Localizing CDFs}
\label{subsec:hp-refhp-motivation}

H-P sort is an integer-sorting algorithm based on histograms and prefix sums~\cite{eisenstat2007sumcrcw,kozakai2021hp}.
Kozakai et al. position H-P sort and its extensions as counting-sort-based methods, and we likewise treat H-P sort as one implementation family of counting-sort-based integer sorting.
H-P sort is not merely a direct GPU translation of ordinary counting sort; its distinctive feature is that it reconstructs output values from CDF boundary information.
Let $H_1$ be the histogram of the input array and $P_1$ its prefix-sum array.
Because $P_1[v]$ is the number of elements no greater than $v$, $P_1$ can be regarded as a CDF that gives the right endpoint of the output interval for each value.
H-P sort constructs another histogram of $P_1$ and applies a prefix sum to that histogram to obtain the sorted output.
Algorithm~\ref{alg:hp-sort} shows this procedure.

\begin{algorithm}[H]
\caption{Basic structure of H-P sort}
\label{alg:hp-sort}
\begin{algorithmic}[1]
\Require Input array $A[0..n-1]$ and range size $R$
\Ensure Sorted array $B[0..n-1]$
\State $H_1[0..R-1] \gets$ histogram of $A$
\State $P_1[0..R-1] \gets$ prefix sums of $H_1$
\State $H_2[0..n] \gets$ histogram of $P_1$
\State $B[0..n-1] \gets$ prefix sums of $H_2[0..n-1]$
\State \Return $B$
\end{algorithmic}
\end{algorithm}

Algorithm~\ref{alg:hp-sort} can be interpreted as follows.
Because $P_1[v]$ is the right endpoint of the output interval of value $v$, the positions at which values of $P_1$ occur indicate boundaries at which the value in the output array changes.
The array $H_2$ counts these boundary positions, and its prefix sums give the number of boundaries up to each output position $i$.
This number is precisely the value that should be written at output position $i$.
For example, if values 0, 1, and 2 occur 2, 3, and 1 times, respectively, then $P_1=(2,5,6)$.
The boundary positions of $H_2$ are 2, 5, and 6, and the prefix sums of $H_2[0..5]$ are $(0,0,1,1,1,2)$, yielding the sorted array.
Thus, H-P sort reconstructs output values from CDF boundary information.
Eisenstat's integer-sorting algorithm, on which H-P sort is based, runs in $O(\log^* n)$ time on a Sum-CRCW PRAM~\cite{eisenstat2007sumcrcw}.
On a GPU, however, histogram updates for identical values cause atomic contention, so the theoretical parallelism of the PRAM algorithm is not directly reflected in practical performance.

Kozakai et al. implemented H-P sort on a GPU and compared it with CUB~\cite{kozakai2021hp}.
For $n$ from $10^6$ to $10^7$, $minVal=0$, and $maxVal=n/50$, H-P sort was up to 2.97 times faster than CUB.
They also proposed 1-H-P sort, which omits part of H-P sort for inputs without duplicates, and 0-compressed H-P sort, which compresses intervals of absent values in the histogram.
They reported that 1-H-P sort was up to 2.01 times faster than CUB for distinct inputs, and that 0-compressed H-P sort was up to 2.73 times faster when the number of distinct values actually present was small~\cite{kozakai2021hp}.
An important result of that work is that H-P sort is not merely a theoretical algorithm: it identifies conditions under which H-P sort is competitive with CUB for range-restricted integer sorting on GPUs.
At the same time, its performance varies substantially with the input size $n$, the range size $R$, and the number of distinct values that actually occur.
In particular, when many threads concurrently add to the same histogram bin, atomic contention increases and the first histogram computation becomes a bottleneck.

RefHP, proposed by Takase et al., refines H-P sort to mitigate this GPU implementation bottleneck~\cite{takase2024hpcasia}.
RefHP achieved a maximum speedup of 3.45 over the original H-P sort and outperformed CUB even under conditions where H-P sort was slower than CUB.
When constructing the CDF, RefHP extends the bin for each value into multiple bins and distributes input elements among these extended bins.
If the extension factor is $e$, it allocates $e$ bins for value $v$ and adds the element at input position $i$ to a location such as $ve+(i\bmod e)$.
Consequently, even when a value appears many times, updates are less likely to concentrate on a single bin.
A prefix sum is then computed over the extended histogram; by referring to the last extended bin of each value, the original CDF can be recovered.
This refinement is characterized by distributing atomic contention in global memory through bin extension.

Unlike the original H-P sort in Algorithm~\ref{alg:hp-sort}, RefHP also provides an implementation that does not construct the second histogram~\cite{takase2024hpcasia}.
Let $H_1[0..R-1]$ be the input histogram and let $P_1[0..R-1]$ be the CDF obtained by computing its prefix sums.
For $d>0$,
\[
P_1[d]-P_1[d-1]=H_1[d].
\]
Thus, $P_1[d]>P_1[d-1]$ if and only if value $d$ occurs in the input array.
When $d$ occurs, $P_1[d-1]$ is the number of input elements smaller than $d$, and hence is the starting position of $d$ in the sorted array.
The output array $B$ of length $n$ is therefore initialized to zero, and for every value $d>0$ at which the CDF has a step, RefHP writes
\[
B[P_1[d-1]]\gets d.
\]
At this point, each value is stored only at the first position of its output interval, while all other entries remain zero.
Finally, the prefix maximum of $B$ is computed:
\[
B[i]=\max_{0\le j\le i}B[j].
\]
This propagates each value up to the position immediately preceding the start of the next value's interval and yields the sorted array $B$.
Thus, the construction and prefix sum of the second histogram in Algorithm~\ref{alg:hp-sort} are replaced by Algorithm~\ref{alg:ref-hp}.

\begin{algorithm}[H]
\caption{Output generation in RefHP using CDF step detection and a prefix maximum~\cite{takase2024hpcasia}}
\label{alg:ref-hp}
\begin{algorithmic}[1]
\Require CDF $P_1[0..R-1]$ and input size $n$
\Ensure Sorted array $B[0..n-1]$

\State $B[0..n-1] \gets 0$

\ForAll{$d\in[1,R)$ \textbf{in parallel}}
  \If{$P_1[d]>P_1[d-1]$}
    \State $B[P_1[d-1]]\gets d$
  \EndIf
\EndFor

\State \Call{InclusiveMaximumInPlace}{$B$}
\State \Return $B$
\end{algorithmic}
\end{algorithm}

This procedure is an important improvement for efficiently generating the H-P output on a GPU and corresponds to the counting-sort idea of producing an output interval for every value.
Taken together, the counting sort of Kolonias et al., H-P sort, and RefHP can all be viewed as methods that construct a CDF from a histogram in some form and determine output intervals from that CDF.

From our perspective, the important property of H-P sort and RefHP is that the CDF obtained from the first histogram and its prefix sums determines the output array.
Accordingly, the central design choice in our method is not to handle one CDF for the entire range, but to construct a local CDF for each range tile and retain only the relevant tile's local CDF in shared memory during output generation.

GPU histogram computations commonly reduce atomic contention by constructing multiple partial histograms and merging them afterward.
Sakharnykh presented a two-stage method in which each thread block constructs a partial histogram for its assigned input and the partial histograms are subsequently merged into the final histogram.
In particular, shared-memory partial histograms were shown to be effective on Maxwell-generation GPUs~\cite{sakharnykh2015histogram}.
Nugteren et al. considered storing independent histograms for each warp or thread in shared memory and demonstrated the trade-off between reducing contention and increasing shared-memory consumption.
They also showed that bank conflicts can arise depending on the number of threads and the shared-memory layout, limiting the number of thread blocks that can execute concurrently~\cite{nugteren2011histogram}.
Henriksen et al. studied generalized histogram computations in which bin updates are extended beyond integer addition.
They considered the trade-off between contention reduction through multiple partial histograms and memory consumption, as well as a method that partitions the bin range and scans the input multiple times when the histogram does not fit in shared memory~\cite{henriksen2020generalizedhistogram}.
Although these studies do not directly address H-P sort, they provide important background for a design that handles local CDFs in shared memory.

We partition the range into tiles of width $T$ that fit in shared memory and construct a histogram and a local CDF for each tile.
The local CDF represents output intervals for values within the tile; adding an offset equal to the number of elements smaller than the tile yields positions in the complete output array.
This structure preserves the idea of existing methods that determine output intervals from a histogram and prefix sums while allowing a CDF to be used locally in shared memory.
Our objective is to determine the range sizes, input sizes, and input distributions for which this range tiling is effective relative to CUB, RefHP, and an implementation based on the algorithm of Kolonias et al.

%% file: sections/03_method.tex
\section{Proposed Method: Range-Tiled CDF Sort}
\label{sec:method}

This section describes the proposed algorithm and its GPU implementation.
The algorithm is an unstable integer-sorting method based on counting sort that partitions the value range into multiple tiles.
Rather than handling a CDF for the entire range at once, it constructs a local CDF for each tile that fits in shared memory and generates the output array from these local CDFs.
As in the preceding section, input values are normalized by subtracting $\minVal$.
Let $A[0..n-1]$ be the input array, $B[0..n-1]$ the output array, $R$ the range size, and $T$ the tile width.
The number of tiles is $m=\lceil R/T\rceil$.
For tile $t$, let its starting value be $b_t=tT$ and its width be $w_t=\min(T,R-b_t)$.
Tile $t$ is responsible for the range interval $[b_t,b_t+w_t)$.

\subsection{Algorithm}
\label{subsec:algorithm}

The proposed algorithm consists of four operations: histogram construction for each range tile, local-CDF construction, tile-offset construction, and output generation from the local CDFs.
Algorithm~\ref{alg:tiledcs-overview} gives an overview.
The array $H[0..R-1]$ in global memory initially stores the histogram of each tile and, after the prefix-sum computation, stores the local CDFs.
The array $O[0..m-1]$ in global memory stores the output starting position of each tile.

\begin{algorithm}[H]
\caption{Overview of RT-CDF}
\label{alg:tiledcs-overview}
\begin{algorithmic}[1]
\Require Input array $A[0..n-1]$, range size $R$, and tile width $T$
\Ensure Sorted array $B[0..n-1]$
\State $m \gets \lceil R/T\rceil$
\State $H[0..R-1] \gets 0$
\ForAll{$t \in [0,m)$ \textbf{in parallel}}
  \State \Call{BuildTileHistogram}{$A,H,n,b_t,w_t$}
\EndFor
\ForAll{$t \in [0,m)$ \textbf{in parallel}}
  \State \Call{InclusiveScanInPlace}{$H[b_t..b_t+w_t-1]$} \Comment{Construct a local CDF}
\EndFor
\State \Call{BuildTileOffsets}{$H,O,R,T,m$}
\For{$t\in[0,m)$}
  \State \Call{GenerateOutputFromLocalCDF}{$H,O,B,t,b_t,w_t$}
\EndFor
\State \Return $B$
\end{algorithmic}
\end{algorithm}

Algorithm~\ref{alg:tiledcs-phase1} shows histogram construction for one range tile.
Each thread block has an array $S[0..w-1]$ of tile width $w$ in shared memory.
After initializing $S$ to zero, the threads scan multiple elements of the input array $A$.
Only when input value $A[i]$ belongs to the current tile $[b,b+w)$ do they perform \texttt{atomicAdd} on $S[A[i]-b]$.
The histogram obtained within each block in shared memory is finally added to $H[b..b+w-1]$ in global memory.

After histogram construction is complete, prefix sums are computed over $H[b..b+w-1]$, and the result is written back to the same region.
This produces the local CDF for the range tile.

\begin{algorithm}[H]
\caption{Histogram construction for one range tile}
\label{alg:tiledcs-phase1}
\begin{algorithmic}[1]
\Procedure{BuildTileHistogram}{$A,H,n,b,w$}
  \State Allocate $S[0..w-1]$ in shared memory
  \State $S[0..w-1] \gets 0$
  \ForAll{$i \in [0,n)$ \textbf{in parallel}}
    \State $x \gets A[i]-b$
    \If{$0 \leq x < w$}
      \State \texttt{atomicAdd}$(S[x],1)$
    \EndIf
  \EndFor
  \ForAll{$x \in [0,w)$ \textbf{in parallel}}
    \State \texttt{atomicAdd}$(H[b+x],S[x])$
  \EndFor
\EndProcedure
\end{algorithmic}
\end{algorithm}

Once the local CDFs have been constructed, the number $N_t$ of elements in tile $t$ is given by the final local-CDF entry $H[b_t+w_t-1]$.
Because the tiles are ordered by value, prefix sums of the numbers of elements in the tiles give the output starting position $O[t]$ of tile $t$.
Algorithm~\ref{alg:tiledcs-offset} shows this operation.
For the reason explained in Section~\ref{subsec:gpu-implementation}, Algorithm~\ref{alg:tiledcs-offset} constructs the offsets sequentially.

\begin{algorithm}[H]
\caption{Construction of tile offsets}
\label{alg:tiledcs-offset}
\begin{algorithmic}[1]
\Procedure{BuildTileOffsets}{$H,O,R,T,m$}
  \State $offset \gets 0$
  \For{$t=0$ to $m-1$}
    \State $O[t] \gets offset$
    \State $offset \gets offset + H[b_t+w_t-1]$
  \EndFor
\EndProcedure
\end{algorithmic}
\end{algorithm}

Output generation uses only the local CDFs and does not rescan the input array.
Algorithm~\ref{alg:tiledcs-phase2} shows the processing of one tile.
First, the local CDF $H[b..b+w-1]$ in global memory is loaded into $C[0..w-1]$ in shared memory.
For each output position $p$ within the tile, a binary search finds the smallest $x$ satisfying $C[x]\ge p+1$, and $b+x$ is written to $B[O[t]+p]$.

\begin{algorithm}[H]
\caption{Output generation from a local CDF}
\label{alg:tiledcs-phase2}
\begin{algorithmic}[1]
\Procedure{GenerateOutputFromLocalCDF}{$H,O,B,t,b,w$}
  \State Allocate $C[0..w-1]$ in shared memory
  \State $C[0..w-1] \gets H[b..b+w-1]$
  \State $tileStart \gets O[t]$
  \State $tileCount \gets C[w-1]$
  \ForAll{$p \in [0,tileCount)$ \textbf{in parallel}}
    \State $x \gets \min\{j \mid C[j] \geq p+1\}$
    \State $B[tileStart+p] \gets b+x$
  \EndFor
\EndProcedure
\end{algorithmic}
\end{algorithm}

Like RefHP, the proposed algorithm determines output intervals from a CDF.
During output generation, however, it does not compute a prefix maximum over an array of length $n$; instead, it loads a local CDF into shared memory and directly determines output values.
When the number of tiles is small, the costs of per-tile kernel launches and binary searches are relatively low, and avoiding a prefix maximum over the entire output array provides a substantial advantage.
Histogram construction, however, scans the input array once for every tile, so the scanning overhead increases with the number of tiles $m$.
Because each tile scans the entire input array during histogram construction, the amount of data read from the input array increases with the number of tiles.
This is the principal performance limitation of the proposed algorithm.

We next give an example of local CDFs and tile offsets.
Let the input array be $A=(3,1,2,1,3,0,1)$, the range size be $R=4$, and the tile width be $T=2$.
Tile 0 handles values 0 and 1, whereas tile 1 handles values 2 and 3.
The histogram, local CDF, and tile offset of each tile are shown in Table~\ref{tab:tiled-example}.
For tile 0, the local CDF is $(1,4)$, so value 0 is written at within-tile position $p=0$ and value 1 at positions $p=1,2,3$.
For tile 1, the offset is 4 and the local CDF is $(1,3)$, so value 2 is written at global position 4 and value 3 at positions 5 and 6.
The resulting output is therefore $(0,1,1,1,2,3,3)$.
As this example illustrates, the proposed algorithm does not compute the rank of every input element; it determines the output interval of each value from a local CDF and writes values through parallel processing over output positions.

\begin{table}[H]
\centering
\small
\caption{Example of local CDFs and tile offsets}
\label{tab:tiled-example}
\begin{tabular}{c c c c c}
\toprule
Tile & Values & Histogram & Local CDF & Offset \\
\midrule
0 & 0,1 & $(1,3)$ & $(1,4)$ & 0 \\
1 & 2,3 & $(1,2)$ & $(1,3)$ & 4 \\
\bottomrule
\end{tabular}
\end{table}

\subsection{Correctness and Complexity Analysis}
\label{subsec:correctness-analysis}

Correctness follows from the facts that a local CDF correctly represents output intervals within a tile, tile offsets correctly position the tiles relative to one another, and output generation uniquely selects the value corresponding to each output position.
Let $c_v$ be the number of occurrences of value $v$.
For tile $t$, let $b_t$ be its starting value, $w_t$ its width, and $C_t[0..w_t-1]$ its local CDF.
For convenience, define $C_t[-1]=0$.

\begin{lemma}[Correctness of the local CDF]
\label{lem:local-cdf}
After local-CDF construction, for every tile $t$ and every $x\in[0,w_t)$,
\[
  C_t[x] = \sum_{j=0}^{x} c_{b_t+j}.
\]
\end{lemma}

\begin{proof}
During histogram construction, one is added to bin $x$ of the tile histogram if and only if an input value $A_i$ belongs to tile $t$, that is, if there exists an $x$ such that $A_i=b_t+x$.
Therefore, before the prefix-sum computation, entry $x$ of the tile histogram is $c_{b_t+x}$.
The prefix-sum computation that constructs the local CDF consequently makes entry $x$ equal to the sum of the occurrence counts from entry 0 through entry $x$.
The lemma follows.
\end{proof}

\begin{lemma}[Correctness of tile offsets]
\label{lem:tile-offset}
The value $O[t]$ obtained by Algorithm~\ref{alg:tiledcs-offset} is equal to the number of input elements whose values are smaller than the range handled by tile $t$; that is,
\[
  O[t] = \sum_{u=0}^{b_t-1} c_u.
\]
\end{lemma}
\begin{proof}
Let $N_s$ be the number of elements in tile $s$.
By Lemma~\ref{lem:local-cdf}, $N_s=C_s[w_s-1]$.
Algorithm~\ref{alg:tiledcs-offset} computes $O[t]$ as the prefix sum of $N_s$ over tiles 0 through $t-1$.
Because the tiles handle disjoint intervals in value order, this prefix sum equals the total number of occurrences of values from 0 through $b_t-1$.
The lemma follows.
\end{proof}

\begin{lemma}[Output interval of each value]
\label{lem:value-interval}
Value $b_t+x$ occupies the interval
\[
  [O[t]+C_t[x-1],\ O[t]+C_t[x])
\]
in the output array.
These intervals are mutually disjoint and ordered by value, and the interval for $b_t+x$ has length $c_{b_t+x}$.
\end{lemma}
\begin{proof}
By Lemma~\ref{lem:local-cdf}, the number of values within tile $t$ that are smaller than $b_t+x$ is $C_t[x-1]$, and the number no greater than $b_t+x$ is $C_t[x]$.
By Lemma~\ref{lem:tile-offset}, $O[t]$ is the number of elements whose values are smaller than the range handled by tile $t$.
Thus, in the complete output array, the number of values smaller than $b_t+x$ is $O[t]+C_t[x-1]$, and the number no greater than $b_t+x$ is $O[t]+C_t[x]$.
The claimed interval follows.
Its length is $C_t[x]-C_t[x-1]=c_{b_t+x}$, and the monotonicity of the CDF together with the construction of tile offsets implies that the intervals for distinct values occur consecutively in value order.
\end{proof}

\begin{theorem}[Correctness of the output array]
\label{thm:rtcdf-correctness}
After output generation, $B$ contains the elements of $A$ in nondecreasing order.
\end{theorem}
\begin{proof}
For each within-tile position $p$, output generation finds the smallest $x$ satisfying $C_t[x]\ge p+1$.
This operation selects the unique $x$ such that $p\in[C_t[x-1],C_t[x])$.
Therefore, exactly $c_{b_t+x}$ output positions in the interval for value $b_t+x$ given by Lemma~\ref{lem:value-interval} are filled with $b_t+x$.
The intervals for all values are mutually disjoint, and their total length is $\sum_{v=0}^{R-1}c_v=n$, so the entire output array is filled exactly once.
Because the intervals are arranged in value order, the output array is in nondecreasing order.
The theorem follows.
\end{proof}

We next summarize the computational complexity.
GPU execution time depends strongly on memory bandwidth, atomic contention, occupancy, and the number of kernel launches, and therefore cannot be characterized adequately by a single parallel-time complexity as in a PRAM model.
We instead describe the total work and workspace of the algorithm, and evaluate actual performance experimentally in Section~\ref{sec:experimental-evaluation}.

Let $m=\lceil R/T\rceil$ be the number of tiles.
During histogram construction, every input element is read for every tile to test whether it belongs to that tile, yielding $O(mn)$ work.
The number of actual additions to shared-memory histograms over all tiles is nevertheless $O(n)$, because each input element belongs to exactly one tile.
Merging the shared-memory histogram of each thread block into global memory requires work proportional to the tile width and the number of blocks.
The total work of the within-tile prefix sums is $O(R)$ because each bin is processed a constant number of times over all tiles.
Constructing tile offsets takes $O(m)$ work.
Output generation performs a binary search in a local CDF for every output position and therefore makes $O(n\log T)$ comparisons.
The simplified total work is thus
\[
  O(mn+n\log T+R).
\]
In practice, this is supplemented by kernel-launch costs, global-memory merging proportional to the number of blocks, and atomic-operation costs at the different levels of the memory hierarchy.\footnote{Each tile also merges block-local histograms into global memory. If $G$ thread blocks are assigned to one range tile, this entails bin-merging work equivalent to $O(mGT)$. Both $T$ and $G$ are implementation parameters, and this cost is included in the measurements.}

Excluding the input and output arrays, the primary global-memory workspaces are $H[0..R-1]$, which stores local CDFs, and $O[0..m-1]$, which stores tile offsets.
The workspace is therefore $O(R+m)$.
Furthermore, with 32-bit integers, the histogram-construction and output-generation kernels use $4T$ bytes of dynamic shared memory.

This analysis identifies the effective range of the implementation.
When $R$ is small and $m$ is one or a small number, the number of input scans is low and handling local CDFs in shared memory provides a substantial advantage.
As $R$ grows, $m$ increases and histogram construction scans the input array repeatedly.
The $O(mn)$ term then becomes dominant, reducing the advantage over radix sort and RefHP.
Accordingly, RT-CDF is positioned as an improvement to counting-sort- and H-P-sort-based implementations for integer arrays with sufficiently small ranges.

\subsection{GPU Implementation}
\label{subsec:gpu-implementation}

We implemented RT-CDF in CUDA C++.
The operations in Algorithm~\ref{alg:tiledcs-overview} are executed by four kernel types: histogram construction for range tiles, prefix sums for constructing local CDFs, range-tile offset construction, and output generation from local CDFs.
We next describe the GPU implementation of each operation.

The histogram-construction kernel processes all range tiles in parallel in one kernel launch.
We fix the block size at 512 threads.
Given a dynamic shared-memory allocation of $4T$ bytes, let $a_{\mathrm{hist}}$ be the number of thread blocks that can reside concurrently on one SM, as determined from the kernel's resource usage.\footnote{The implementation uses \texttt{cudaOccupancyMaxActiveBlocksPerMultiprocessor} from the CUDA Runtime API.}
If the GPU has $S$ SMs, the number of thread blocks assigned to each range tile is
\[
g_{\mathrm{hist}}
=
\min\left\{
\max\left(1,\left\lfloor\frac{a_{\mathrm{hist}}\cdot S}{m}\right\rfloor\right),
\left\lceil\frac{n}{512}\right\rceil
\right\},
\]
and the grid size of the entire kernel is
\[
G_{\mathrm{hist}}=m\cdot g_{\mathrm{hist}}.
\]
Here, $a_{\mathrm{hist}}\cdot S$ is the number of thread blocks that can reside concurrently on the entire GPU after accounting for the kernel's resource use; these blocks are divided evenly among the $m$ range tiles.

Each thread block derives its assigned range tile $t$ and its block number $q$ within that tile from the block index as
\[
t=\texttt{blockIdx.x}\bmod m,\qquad
q=\left\lfloor\frac{\texttt{blockIdx.x}}{m}\right\rfloor.
\]
A thread assigned to range tile $t$ starts at input position
\[
i=512\cdot q+\texttt{threadIdx.x}
\]
and cooperatively scans the input array using a grid-stride loop that advances by $512\cdot g_{\mathrm{hist}}$ positions per iteration.
Only when the loaded value belongs to the target range tile $[b_t,b_t+w_t)$ does the thread perform \texttt{atomicAdd} on the corresponding bin in dynamic shared memory.
After scanning the input array, each thread block adds the partial histogram obtained in shared memory to $H[b_t..b_t+w_t-1]$ in global memory.
Thus, rather than performing a global-memory atomic addition for every input element, each thread block adds its aggregate count for each bin to global memory once.

After histogram construction, the local-CDF construction kernel computes the prefix sums of each range-tile histogram and writes the result back to the same region of $H$.
This kernel is launched once with $m$ thread blocks of 1024 threads each, with one block assigned to each range tile.
The block assigned to tile $t$ partitions $H[b_t..b_t+w_t-1]$ among its threads.
Each thread first computes partial sums over its contiguous subrange, then computes the sum of all preceding subranges by applying a block-level prefix sum to the per-thread totals, and adds that sum to every element of its own subrange.
This constructs the local CDFs for all range tiles in parallel in one kernel launch.

The range-tile offset kernel is launched with one thread block containing one thread.
It processes the range tiles in order and obtains the number of input elements in each tile from the final entry of its local CDF.
It sequentially accumulates the number of elements appearing before each tile and stores the output starting position of tile $t$ in $O[t]$.
Because the number of repeated input scans during histogram construction grows with the number of tiles $m$, the implementation is intended for ranges in which $m$ is relatively small.
Consequently, the $O(m)$ tile-offset construction has little effect on total execution time and is performed sequentially by one thread.

For output generation, range tiles are not processed concurrently; instead, a kernel is launched sequentially for each tile containing output elements.
Let $N_t$ be the number of input elements in range tile $t$.
The value $N_t$ is needed both to decide whether to launch the kernel and to determine the number of thread blocks in the launch.

When $m=1$, $N_0=n$ and this value is used directly.
When $m\ge2$, the range-tile offset array $O$ of size $m$ is copied once from the device to the host, and the host computes
\[
N_t=
\begin{cases}
O[t+1]-O[t] & (t<m-1),\\
n-O[t] & (t=m-1).
\end{cases}
\]
No output-generation kernel is launched for a tile with $N_t=0$.

The output-generation kernel uses a fixed block size of 512 threads.
As for histogram construction, let $a_{\mathrm{out}}$ be the number of blocks that can reside concurrently on one SM when the kernel uses $4T$ bytes of dynamic shared memory.
If the GPU has $S$ SMs, the number of thread blocks launched for a nonempty tile $t$ is
\[
g_{\mathrm{out}}(t)
=
\min\left\{
a_{\mathrm{out}}\cdot S,\,
\left\lceil\frac{N_t}{512}\right\rceil
\right\}.
\]
Unlike histogram construction, one output-generation launch processes only one range tile, so the number of blocks that can reside concurrently on the entire GPU is not divided by $m$.
This permits enough thread blocks to be assigned to the output of a tile even when input elements are concentrated in only a few tiles, reducing load imbalance caused by unequal tile populations.

Each thread block loads the local CDF of its target range tile from global memory into dynamic shared memory.
Each thread starts at within-tile output position
\[
i=512\cdot\texttt{blockIdx.x}+\texttt{threadIdx.x}
\]
and cooperatively handles the $N_t$ output positions in the tile using a grid-stride loop that advances by $512\cdot g_{\mathrm{out}}(t)$ positions per iteration.
For each output position, the thread performs a binary search in the local CDF in shared memory and writes the corresponding value to the output array $B$.
Output generation therefore does not reread the input array and distributes the output positions within a tile among the threads.

Table~\ref{tab:rtcdf-kernel-config} summarizes the launch configurations of these kernels.

\input{tables/table_rtcdf_kernel_config}
\FloatBarrier

The histogram-construction and output-generation kernels use $4T$ bytes of dynamic shared memory to store tile-local arrays.
If the dynamic shared-memory requirement for tile width $T$ exceeds the default shared-memory limit per thread block, the implementation uses CUDA's opt-in configuration to increase the permitted dynamic shared-memory allocation up to the GPU's supported limit.\footnote{The implementation sets this option using \texttt{cudaFuncSetAttribute} from the CUDA Runtime API.}
The tile widths used in the evaluation and their selection procedure are described in Section~\ref{subsec:implementation-selection}.

%% file: tables/table_rtcdf_kernel_config.tex
\begin{table}[H]
\centering
\caption{Kernel launch configurations in RT-CDF}
\label{tab:rtcdf-kernel-config}
\begin{tabular}{lccc}
    \toprule
    Operation
      & threads / block
      & blocks / grid \\
    \midrule
    Histogram construction
      & 512
      & $m \cdot g_{\mathrm{hist}}$ \\
    Local-CDF construction
      & 1024
      & $m$ \\
    Tile-offset construction
      & 1
      & 1 \\
    Output generation
      & 512
      & $g_{\mathrm{out}}(t)$ \\
    \bottomrule
\end{tabular}
\end{table}

%% file: sections/04_experimental_evaluation.tex
\section{Experimental Evaluation}
\label{sec:experimental-evaluation}

This section evaluates the RT-CDF implementation on an NVIDIA GeForce RTX 4090.
Section~\ref{subsec:implementation-selection} first examines the performance characteristics of the output-generation methods considered for RT-CDF and uses preliminary experiments to select the tile width for the main evaluation.
Section~\ref{subsec:main-evaluation-setup} then describes the baselines and measurement conditions for the main evaluation.
Section~\ref{subsec:overall-performance} compares RT-CDF with CUB, RefHP, and an implementation based on the algorithm of Kolonias et al., and Section~\ref{subsec:discussion-limitations} discusses the results.

\subsection{Preliminary Evaluation of Implementation Choices}
\label{subsec:implementation-selection}

This subsection evaluates the performance characteristics of the output-generation method used by RT-CDF and selects the tile width used in the main evaluation.
The preliminary experiments in this subsection and the main evaluation in Section~\ref{subsec:main-evaluation-setup} were conducted in the same environment, shown in Table~\ref{tab:environment}.
One KiB is 1024 bytes.

\input{tables/table_environment}
\FloatBarrier

Table~\ref{tab:preliminary-measurement} summarizes the conditions of the preliminary evaluation.

\input{tables/table_preliminary_measurement}
\FloatBarrier

We used uniform and normal input distributions.
For the uniform distribution, each element was generated uniformly at random from the integer interval $[0,R-1]$.
For the normal distribution, a real value was generated from a normal distribution with mean $(R-1)/2$ and standard deviation $0.125R$, and was rounded to the nearest integer.
A rounded value outside $[0,R-1]$ was clipped to 0 or $R-1$.

Because every entry of a tile-local histogram and local CDF is a 32-bit integer, a thread block requires at most $4T$ bytes of shared memory to process a tile.
For the largest candidate tile width, $T=24576$, this amount is $4T=96\ \mathrm{KiB}$, which is below the 99-KiB opt-in shared-memory limit per thread block shown in Table~\ref{tab:environment}.
For each range size $R$, we evaluated only candidates satisfying $T\le R$ and whose shared-memory requirement did not exceed the opt-in limit per thread block.
Moreover, because increasing the number of tiles $m$ causes the input array to be scanned for every tile, thereby increasing both the amount of work and the number of global-memory accesses, we restricted the candidates to tile widths satisfying $m=\lceil R/T\rceil\le16$ to avoid excessive rescanning.

To reduce the influence of temporary delays and execution-time variability, we performed multiple independent measurement runs for each condition.
In each run, several warm-up executions preceded the measurements; the median of the measurements was used as the representative value for that run, and the median of the representative values over all runs was used as the final representative value.
Because GPU clock frequencies and system conditions may differ between runs, this aggregation reduces the effect of a temporary fluctuation in any one run.

We use host-side wall-clock time to evaluate the total cost of a single sorting operation.
The measured interval begins with allocation of algorithm-specific workspace and ends after the GPU sorting operation has completed and the workspace has been deallocated.
Allocation of the input and output arrays is common to all methods and is therefore excluded.

We first compared the prefix-maximum method in Algorithm~\ref{alg:ref-hp} with the binary-search method over local CDFs used in Algorithm~\ref{alg:tiledcs-phase2}, in order to examine the performance characteristics of the output-generation method used by RT-CDF.
The two methods shared the same procedure for constructing a histogram from the input.
The binary-search method then constructed a local CDF for every tile and computed tile offsets, whereas the prefix-maximum method constructed a CDF for the complete range; each method subsequently performed its own output-generation procedure.
Thus, this experiment compares the total sorting time, including the CDF construction required by each method, rather than only the output-generation kernels.

Because the performance of both methods varies with tile width, we measured every feasible tile-width candidate for every combination of range size, input size, and input distribution.
For each method, the minimum time over the tile-width candidates was taken as the execution time for that condition.
This compares the best performance obtainable by each method without fixing a tile width that might favor or disadvantage one of them.
Figure~\ref{fig:output-method-heatmap} shows the results.
We define
\[
\rho =
\frac{\text{minimum execution time of the prefix-maximum method}}
     {\text{minimum execution time of the binary-search method}}.
\]
The binary-search method is faster when $\rho>1$, whereas the prefix-maximum method is faster when $\rho<1$.
Each cell in the figure shows $\rho$ for a particular range size $R$, input size $n$, and input distribution.
The displayed values are rounded to two decimal places.

\input{figures/fig_output_method_heatmap}
\FloatBarrier

As shown in Figure~\ref{fig:output-method-heatmap}, the binary-search method was faster in 184 of the 240 conditions.
Its advantage became more pronounced as the input size increased: for all 120 conditions with $n\ge5\times10^7$, the binary-search method was faster regardless of the range size or input distribution.
For small inputs and large ranges, by contrast, the prefix-maximum method was faster.
This appears to be because a large range increases the number of tiles, making fixed overheads of the binary-search method---including per-tile kernel launches and retrieval of tile populations before output generation---relatively large for small inputs.
Our evaluation includes inputs of up to $10^9$ elements and uses one common output-generation algorithm for all conditions.
The binary-search method was consistently faster for large inputs, indicating that this design is suitable for RT-CDF evaluations that include large-scale inputs.
Nevertheless, switching the output-generation method according to the input size and range could further improve RT-CDF under conditions in which binary search was slower.
Such adaptive selection is a possible future improvement.

We next select the tile width $T$ for the binary-search method used in the main evaluation.
Using the binary-search measurements from the preceding comparison of output-generation methods, we compared tile widths separately for each range size $R$.
For each $R$, the comparison covered 20 conditions formed by 10 input sizes $n$ and two input distributions.
For each condition, we define the normalized execution time of tile width $T$ as
\[
r(T)=
\frac{\text{execution time with $T$}}{\text{minimum execution time for that condition}}.
\]
Thus, $r(T)\ge1$, and $r(T)=1$ means that $T$ is the fastest tile width under that condition.
For each tile width, we computed the geometric mean of $r(T)$ over the 20 conditions and compared tile widths using this value.
A value closer to 1 means that the tile width performs closer to the fastest choice for each condition across the 20 combinations of input size and input distribution.
Table~\ref{tab:tile-width-relative-time} presents the results.

\input{tables/table_tile_width_relative}
\FloatBarrier

Based on Table~\ref{tab:tile-width-relative-time}, for each range size $R$ we selected the tile width $T$ with the smallest geometric mean of normalized execution time.
The selected tile widths and corresponding numbers of tiles are shown in Table~\ref{tab:selected-tile-width}.

\input{tables/table_selected_tile_width}
\FloatBarrier

The main evaluation below fixes the tile width for each range size to the value in Table~\ref{tab:selected-tile-width}.

\subsection{Baselines and Measurement Conditions for the Main Evaluation}
\label{subsec:main-evaluation-setup}

This subsection describes the conditions of the main evaluation using the tile widths selected in Section~\ref{subsec:implementation-selection}.
The experimental environment was the same as in the preliminary evaluation.
We compared RT-CDF with CUB, RefHP, and an implementation based on the algorithm of Kolonias et al.
Table~\ref{tab:comparison-targets} summarizes the methods.
For CUB, we used the known range size $R$ to restrict the bit range processed by \texttt{DeviceRadixSort} to $[0,\lceil\log_2R\rceil)$.
Because this setting reduces the number of bits processed relative to sorting all 32 bits, it favors CUB and reflects the known-range problem setting of this study more appropriately than always processing a complete 32-bit integer.
RefHP is a prior implementation that refines H-P sort for GPUs and is the closest baseline to RT-CDF.
We also included an implementation based on Kolonias et al. because its output generation, which produces one interval per value, is closely related to RT-CDF.

\input{tables/table_comparison_targets}

Table~\ref{tab:measurement-method} summarizes the conditions of the main evaluation.

\input{tables/table_measurement_method}

The uniform and normal inputs were generated in the same manner as in the preliminary evaluation.
For all-equal inputs, every element was set to zero.
The range supplied to the algorithms remained $R$, as for the other distributions; it was not changed to $R=1$ even though only one distinct value actually occurred.
This evaluates performance when all values are concentrated at one point while the declared range remains fixed.

Because the output-generation method of Kolonias et al. assigns work by value, concentrating occurrences on one value creates a workload imbalance~\cite{kolonias2011countsort}.
The all-equal input used here is therefore an unfavorable condition for that output-generation method.

The aggregation of measurements and the measured interval were the same as in Section~\ref{subsec:implementation-selection}.
If an execution time exceeded 1000\,ms, measurements for larger input sizes with the same range $R$ and input distribution were terminated.

To investigate the causes of performance changes in RT-CDF, Section~\ref{subsec:discussion-limitations} uses an additional timing experiment on RT-CDF alone.
This experiment fixed the input size at $n=10^8$ and used the same 12 range sizes, three input distributions, and tile widths from Table~\ref{tab:selected-tile-width} as the main evaluation.
The number of measurements and their aggregation also followed the main evaluation: three independent runs were performed, each containing 30 measurements after 10 warm-up executions; the median within each run was used as its representative value, and the median of the three run representatives was used as the final value.
CUDA events recording positions in the GPU execution stream were placed before and after each processing stage, and elapsed times between event pairs were measured.
The measured stages were histogram-array initialization, histogram construction, local-CDF construction, tile-offset construction, and output generation.\footnote{The implementation uses \texttt{cudaEventRecord} and \texttt{cudaEventElapsedTime} from the CUDA Runtime API.}

\subsection{Performance Evaluation}
\label{subsec:overall-performance}

We first fix the input size at $n=10^8$.
Tables~\ref{tab:n1e8-times-uniform}, \ref{tab:n1e8-times-gaussian}, and \ref{tab:n1e8-times-all-equal} show execution times and speedup ratios for the uniform, normal, and all-equal inputs, respectively.
Here, ``Best baseline'' denotes the fastest execution time among CUB, RefHP, and the implementation based on Kolonias et al.
The speedup ratio is the execution time of the fastest baseline divided by the execution time of RT-CDF.
Consequently, a ratio greater than 1 means that RT-CDF is faster than every baseline.

\input{tables/table_n1e8_times_uniform}
\FloatBarrier

\input{tables/table_n1e8_times_gaussian}
\FloatBarrier

\input{tables/table_n1e8_times_all-equal}
\FloatBarrier

For all three input distributions, RT-CDF was the fastest method from $R=2^7$ through $2^{17}$, whereas CUB was fastest at $R=2^{18}$.
For $R\le2^{14}$, the RT-CDF execution times differed little across the input distributions, and all three distributions exhibited the same tendency of increasing execution time with the range size.
Even at $R=2^{17}$, RT-CDF was 1.15, 1.19, and 1.22 times faster than the fastest baseline for the uniform, normal, and all-equal inputs, respectively.
At $R=2^{18}$, its speedup ratios relative to the fastest baseline fell to 0.62, 0.62, and 0.64, respectively, and the advantage of RT-CDF disappeared independently of the input distribution.
The largest ratios in these tables occurred at $R=2^{10}$ for the uniform and all-equal inputs and at $R=2^9$ for the normal input, and were approximately 2.65--2.67.
The implementation based on Kolonias et al. required approximately 267--273\,ms for every range under the all-equal input, much longer than for the other two distributions.
This is presumably because its threads are assigned by value and almost all output generation is concentrated in the single thread assigned to value zero.
CUB was the fastest baseline in every one of these 36 conditions, and RT-CDF outperformed CUB and all other baselines in 33 of the 36.
The only conditions in which RT-CDF was not the fastest were the three input distributions at $R=2^{18}$.
These results show that, for $n=10^8$, RT-CDF performance depends primarily on the range size and only weakly on the input distribution.

We next fix the range and compare execution time as a function of the input size $n$.
As a representative range for which RT-CDF performs well over a broad interval of input sizes, we first consider $R=2^{10}$.
Figures~\ref{fig:runtime-r1024-uniform}, \ref{fig:runtime-r1024-gaussian}, and \ref{fig:runtime-r1024-all-equal} show the uniform, normal, and all-equal results, respectively.

\input{figures/fig_runtime_r1024_uniform}
\FloatBarrier

\input{figures/fig_runtime_r1024_gaussian}
\FloatBarrier

\input{figures/fig_runtime_r1024_all_equal}
\FloatBarrier

For $R=2^{10}$, RT-CDF was fastest at all 28 measured input sizes for each of the three input distributions.
Thus, it outperformed every baseline in all 84 conditions across the three distributions.
The RT-CDF curves for the three input distributions were also nearly identical over the entire measurement range, indicating little distribution-dependent performance variation at $R=2^{10}$.
For all-equal inputs, the execution time of the implementation based on Kolonias et al. increased sharply with the input size, and results beyond the termination threshold are omitted from the figures.

We next consider $R=2^{17}$, near the boundary at which the relative performance of RT-CDF and the baselines changes with the input size.
Figures~\ref{fig:runtime-r131072-uniform}, \ref{fig:runtime-r131072-gaussian}, and \ref{fig:runtime-r131072-all-equal} show the uniform, normal, and all-equal results, respectively.

\input{figures/fig_runtime_r131072_uniform}
\FloatBarrier

\input{figures/fig_runtime_r131072_gaussian}
\FloatBarrier

\input{figures/fig_runtime_r131072_all_equal}
\FloatBarrier

At $R=2^{17}$, unlike at $R=2^{10}$, a baseline was faster for small inputs, whereas RT-CDF became the fastest method as the input size increased.
For the uniform distribution, RefHP was fastest for $n\le10^7$, CUB was fastest at $n=2\times10^7$, and RT-CDF was fastest for $n\ge3\times10^7$.
For the normal distribution, RefHP was fastest at $n=10^6$, CUB for $2\times10^6\le n\le10^7$, and RT-CDF for $n\ge2\times10^7$.
For all-equal inputs, CUB was fastest for $n\le10^7$, whereas RT-CDF was fastest for $n\ge2\times10^7$.
Across the three distributions, RT-CDF was fastest in 53 of the 84 conditions.
Thus, $R=2^{17}$ is a boundary regime in which fixed overheads favor the baselines for small inputs, but the advantage of RT-CDF emerges once the input is sufficiently large.

Finally, we examine the largest evaluated range, $R=2^{18}$.
Figures~\ref{fig:runtime-r262144-uniform}, \ref{fig:runtime-r262144-gaussian}, and \ref{fig:runtime-r262144-all-equal} show execution time as a function of $n$ for the uniform, normal, and all-equal inputs, respectively.

\input{figures/fig_runtime_r262144_uniform}
\FloatBarrier

\input{figures/fig_runtime_r262144_gaussian}
\FloatBarrier

\input{figures/fig_runtime_r262144_all_equal}
\FloatBarrier

For $R=2^{18}$, RT-CDF was not the fastest method at any of the 28 measured input sizes for any distribution.
For uniform and normal inputs, the fastest method alternated between CUB and RefHP depending on the input size, but both outperformed RT-CDF throughout the measurement range.
For all-equal inputs, RT-CDF and CUB were nearly tied at $n=10^6$, but CUB was faster thereafter.
Moreover, increasing the input size did not produce a trend in which RT-CDF caught up with and overtook the baselines; the large-input advantage observed at $R=2^{17}$ was lost at $R=2^{18}$.

In summary, RT-CDF was superior independently of the input size and distribution at $R=2^{10}$; its relative performance changed with the input size at $R=2^{17}$; and the baselines were superior throughout the measurement range at $R=2^{18}$.
The performance boundary of RT-CDF therefore appears clearly between $R=2^{17}$ and $R=2^{18}$.

We next present the speedup ratio defined as the execution time of the fastest baseline divided by that of RT-CDF for each condition.
A ratio greater than 1 means that RT-CDF is fastest; a ratio below 1 means that at least one baseline is faster.
In the version of CUB \texttt{DeviceRadixSort} used in this experimental environment, 32-bit keys are processed using Onesweep, with up to eight bits processed per pass.
We refer to each such processing step as a pass.
Because CUB was instructed to sort only the low-order $\log_2R$ bits required by the range, it requires one pass for $R\le2^8$, two passes for $2^9\le R\le2^{16}$, and three passes for $2^{17}\le R\le2^{18}$.
To account for these differences, we divide the speedup results into these three ranges.
For the two-pass range, we show representative results for $R=2^9,2^{12},2^{15},2^{16}$.
Figures~\ref{fig:speedup-uniform}, \ref{fig:speedup-gaussian}, and \ref{fig:speedup-all-equal} show the uniform, normal, and all-equal results, respectively.
Because the numerator is the fastest baseline for each condition, changes in the curves reflect not only changes in the performance gap relative to RT-CDF but also changes in which baseline is fastest.

\input{figures/fig_speedup_uniform}
\FloatBarrier

For uniform inputs, the largest advantage over the baselines occurred at $R=2^9$, where the maximum speedup was 3.70 at $n=9\times10^6$.
For $R=2^9$ and $2^{12}$, the ratio exceeded 1 at every measured input size.
As the range increased, however, the ratio for small inputs decreased: it fell below 1 at $n=10^6,2\times10^6$ for $R=2^{15}$ and at $n\le5\times10^6$ for $R=2^{16}$.
This tendency became more pronounced at $R=2^{17}$, where a baseline was faster for $n\le2\times10^7$, whereas RT-CDF was fastest for $n\ge3\times10^7$.
At $R=2^{18}$, the ratio was below 1 for every measured input size, with no indication that RT-CDF would overtake the baselines as the input size increased.

\input{figures/fig_speedup_gaussian}
\FloatBarrier

For normal inputs, the largest advantage likewise occurred at $R=2^9$, with a maximum speedup of 3.60 at $n=9\times10^6$.
RT-CDF was fastest throughout the measured range for $R=2^9$ and $2^{12}$.
The ratio fell below 1 at $n=10^6$ for $R=2^{15}$ and at $n\le5\times10^6$ for $R=2^{16}$, but RT-CDF was fastest for larger inputs.
At $R=2^{17}$, a baseline was faster for $n\le10^7$, whereas RT-CDF was fastest for $n\ge2\times10^7$.
As for the uniform distribution, the ratio was below 1 for every input size at $R=2^{18}$.

\input{figures/fig_speedup_all_equal}
\FloatBarrier

For all-equal inputs, the maximum speedup observed in the entire experiment, 4.39, occurred at $R=2^9$ and $n=10^6$.
For each representative two-pass range shown in the figure---$R=2^9,2^{12},2^{15},2^{16}$---the ratio exceeded 1 at every measured input size.
This differs from the uniform and normal results for large ranges, in which the ratio fell below 1 for small inputs.
At $R=2^{17}$, a baseline was faster for $n\le10^7$, whereas RT-CDF was fastest for $n\ge2\times10^7$.
At $R=2^{18}$, RT-CDF and CUB were nearly tied at $n=10^6$, but the ratio was at most 1 for every measured input size.

For all three input distributions, the largest advantage of RT-CDF over the baselines was obtained at $R=2^9$.
The ratio generally decreased as the range grew, and a baseline was particularly likely to be superior for small inputs with a large range.
At $R=2^{17}$, all three distributions exhibited a boundary behavior in which the ratio crossed 1 as the input size increased; at $R=2^{18}$, RT-CDF was never fastest over the measured range.
The speedup ratios therefore also confirm that the performance boundary lies between $R=2^{17}$ and $R=2^{18}$.

We next identify the fastest method at representative input sizes.
Tables~\ref{tab:winner-map-uniform}, \ref{tab:winner-map-gaussian}, and \ref{tab:winner-map-all-equal} show the method with the smallest execution time for each range at $n=10^6,10^7,10^8,10^9$.

\input{tables/table_winner_map_uniform}
\FloatBarrier

\input{tables/table_winner_map_gaussian}
\FloatBarrier

\input{tables/table_winner_map_all_equal}
\FloatBarrier

For all three distributions, RT-CDF was fastest over a broad interval of input sizes when the range was small.
As the range grew, RefHP or CUB was increasingly likely to be fastest, particularly for small inputs.
At $R=2^{17}$, a baseline was fastest at $n=10^6$ and $10^7$ for all three distributions, whereas RT-CDF was fastest at $n=10^8$ and $10^9$.
This confirms that the relative performance of RT-CDF and the baselines changes with the input size at $R=2^{17}$.
At $R=2^{18}$, RT-CDF was not fastest for any of the four representative input sizes.
For uniform and normal inputs, RefHP or CUB was fastest depending on the input size, and CUB was fastest in all four all-equal conditions.
Thus, these winner maps show the same pattern as the execution-time and speedup results: RT-CDF retains an advantage for large inputs at $R=2^{17}$ but loses it at $R=2^{18}$.
Among all 144 combinations of representative input size and range, RT-CDF was fastest in 119, RefHP in 13, and CUB in 12; the implementation based on Kolonias et al. was never fastest.

Finally, we compare the absolute execution times of all methods at representative input sizes.
For $n=10^6,10^7,10^8,10^9$, Tables~\ref{tab:runtime-four-n-uniform}, \ref{tab:runtime-four-n-gaussian}, and \ref{tab:runtime-four-n-all-equal} show the execution times of RT-CDF, CUB, RefHP, and the implementation based on Kolonias et al. over every evaluated range for uniform, normal, and all-equal inputs, respectively.
The smallest execution time in each condition is shown in bold.
These tables make it possible to examine not only which method is fastest, but also the absolute magnitude of the performance differences.

\input{tables/table_runtime_four_n_uniform}
\FloatBarrier

\input{tables/table_runtime_four_n_gaussian}
\FloatBarrier

\input{tables/table_runtime_four_n_all-equal}
\FloatBarrier

As Tables~\ref{tab:runtime-four-n-uniform}--\ref{tab:runtime-four-n-all-equal} show, when RT-CDF outperforms the baselines on large inputs, the absolute reduction in execution time is also substantial.
For example, for a uniform input with $R=2^{15}$ and $n=10^9$, RT-CDF runs in 11.88\,ms, reducing the execution time by approximately 13.12\,ms relative to the fastest baseline, CUB, at 25.00\,ms.
For the same range and input size under the normal distribution, RT-CDF runs in 12.14\,ms and CUB in 24.87\,ms, a difference of approximately 12.73\,ms.

For small inputs, by contrast, even a large speedup ratio corresponds to a small absolute difference.
For an all-equal input with $R=2^9$ and $n=10^6$, the speedup is approximately 4.39, but RT-CDF runs in 0.02140\,ms and the fastest baseline, CUB, in 0.09403\,ms, a difference of approximately 0.073\,ms.
Both the ratio and the absolute execution time should therefore be considered.

As the range approaches the performance boundary, the absolute advantage on large inputs also narrows.
For a uniform input with $R=2^{17}$ and $n=10^9$, RT-CDF runs in 28.55\,ms and the fastest baseline, RefHP, in 29.75\,ms, an advantage of approximately 1.20\,ms.
At $R=2^{18}$ and $n=10^9$, RT-CDF takes 53.14\,ms, whereas RefHP takes 29.75\,ms, making RT-CDF approximately 23.39\,ms slower.
This again demonstrates that the performance boundary occurs between $R=2^{17}$ and $R=2^{18}$.

\subsection{Execution-Time Breakdown and Performance Limitations}
\label{subsec:discussion-limitations}

To investigate why RT-CDF performance changes with the range and input distribution, we used the additional measurements described in Section~\ref{subsec:main-evaluation-setup} to break down the time spent in each processing stage.
Figure~\ref{fig:rtcdf-kernel-breakdown} shows the results with the input size fixed at $n=10^8$.
The figure presents the execution times of two processing stages: histogram construction and output generation.
We also measured local-CDF construction, tile-offset construction, and histogram-array initialization, but each required no more than 0.037\,ms under every condition and was sufficiently small relative to histogram construction and output generation to be omitted from the figure.
The breakdown reported here is based on elapsed times between CUDA events, as described in Section~\ref{subsec:main-evaluation-setup}, and is not a complete decomposition of the host-side wall-clock time used in the main evaluation.
The main-evaluation time includes allocation and deallocation of method-specific workspace, whereas the times in this figure do not.

\input{figures/fig_rtcdf_kernel_breakdown}
\FloatBarrier

Under the settings in Table~\ref{tab:selected-tile-width}, the number of range tiles is $m=1$ for $R\le2^{12}$.
In this interval, histogram-construction time and output-generation time are similar, and neither stage differs substantially across the three input distributions.

For a fixed range, histogram-construction times differ little among the three distributions, indicating that the effect of skewed input values is limited.
They increase substantially with the range, however: approximately 0.67\,ms at $R=2^{15}$, 1.15--1.18\,ms at $R=2^{16}$, 2.25--2.29\,ms at $R=2^{17}$, and 4.54--4.62\,ms at $R=2^{18}$.

During RT-CDF histogram construction, every input element is tested for membership in each range tile.
With $m$ tiles, the entire input array is therefore examined $m$ times, for $O(mn)$ input references.
Repeated scanning as the number of tiles increases is the principal cause of the higher histogram-construction time for large ranges.
Indeed, from $R=2^{14}$ through $2^{17}$, $m$ increases to 2, 4, 8, and 11, respectively, and the histogram-construction time rises accordingly.

From $R=2^{17}$ to $2^{18}$, histogram-construction time nearly doubles even though $m=11$ in both cases.
As shown in Table~\ref{tab:selected-tile-width}, the tile width is $T=12288$ for $R=2^{17}$ and $T=24576$ for $R=2^{18}$.
Because each thread block retains a histogram of width $T$ in shared memory, the shared-memory use per block with 32-bit integers is approximately 48\,KiB and 96\,KiB, respectively.
Consequently, two thread blocks can reside concurrently on one SM at $R=2^{17}$, but only one at $R=2^{18}$.
The resulting decrease in occupancy due to fewer concurrently resident blocks is one reason that histogram construction becomes slower at $R=2^{18}$ despite the unchanged value of $m$.

Output-generation time changes less with the range than histogram-construction time and also differs little among the input distributions.
For $R\le2^{16}$, it is approximately 0.47--0.53\,ms over all three distributions.
At $R=2^{17}$, it is 0.585\,ms for uniform inputs, 0.581\,ms for normal inputs, and 0.487\,ms for all-equal inputs.
Our implementation launches output-generation kernels sequentially by range tile and chooses the number of thread blocks according to the number of elements in each tile.
As a result, output-generation time does not increase substantially either when elements are distributed across several tiles, as under the uniform distribution, or when they are concentrated in only a few tiles, as under the normal and all-equal inputs.

For an all-equal input, in particular, only the range tile containing value zero is nonempty, and no output-generation kernel is launched for tiles without output elements.
The output-generation time at $R=2^{17}$ is consequently shorter than for the uniform and normal inputs.
This result also shows that the output-generation stage does not suffer a large performance degradation from distribution-induced imbalance in tile populations.

At $R=2^{18}$, output-generation time also increases, to 0.835\,ms for uniform inputs, 0.817\,ms for normal inputs, and 0.668\,ms for all-equal inputs.
The tile width is $T=24576$ at this range, so more shared memory is needed to retain the local CDF during output generation.
Furthermore, for uniform and normal inputs, output-generation kernels must be launched sequentially for several nonempty tiles.
This increase is nevertheless small relative to that of histogram construction.

At $R=2^{18}$, histogram construction accounts for most of the measured time.
For the uniform distribution, histogram construction requires approximately 4.62\,ms and output generation approximately 0.83\,ms; the corresponding values are approximately 4.57 and 0.82\,ms for the normal distribution, and 4.54 and 0.67\,ms for the all-equal input.
Thus, among the principal processing stages measured here, histogram construction dominates at $R=2^{18}$ regardless of the input distribution.

These results are consistent with the main-evaluation trends in Section~\ref{subsec:overall-performance}.
At $R=2^{17}$, the baselines are faster for small inputs, but RT-CDF becomes fastest for sufficiently large inputs under all three distributions.
At $R=2^{18}$, RT-CDF is never fastest over the measured input sizes.
In Figure~\ref{fig:rtcdf-kernel-breakdown}, the increase in histogram-construction time from $R=2^{17}$ to $2^{18}$ is greater than the increase in output-generation time, and this pattern occurs for all three distributions.
The performance loss at $R=2^{18}$ therefore appears to be caused primarily by the cost of constructing histograms for a large range, rather than by a problem specific to one input distribution.

RT-CDF can therefore execute histogram construction and output generation efficiently when the range is relatively small and can be processed with a small number of range tiles.
As the range grows, however, the input array is rescanned more often during histogram construction because the number of tiles $m$ increases.
Moreover, a large tile width increases the shared-memory consumption per thread block, reducing the number of concurrently resident thread blocks and lowering occupancy.
At $R=2^{18}$, these effects make histogram construction dominant and eliminate the advantage of RT-CDF over the baselines.
Extending the applicability of RT-CDF to larger ranges will therefore require reducing histogram-construction cost rather than primarily optimizing output generation.

The implementation based on Kolonias et al. also generates the sorted array from one output interval per value.
Because it assigns output-generation work by value, that work can be distributed among many values under uniform and normal inputs; for an all-equal input, however, almost every output element is generated by the single processing unit corresponding to value zero.
In the main evaluation, its host-side wall-clock time for all-equal inputs of $n=10^8$ was approximately 267--273\,ms throughout the evaluated ranges, and increasing the range did not improve output-generation parallelism.
RT-CDF, in contrast, generates output by range tile and assigns multiple thread blocks according to the number of elements in the tile, allowing output generation to be parallelized even for all-equal inputs.

The primary global-memory workspaces used by RT-CDF, excluding the input and output arrays, are the array $H$ of length $R$ and the tile-offset array $O$ of length $m$.
With 32-bit integers, their total size is $4R+4m$ bytes.
By contrast, the ordinary \texttt{SortKeys} interface of CUB \texttt{DeviceRadixSort} used in this evaluation requires $O(n+P)$ temporary storage for an input of length $n$ and $P$ SMs~\cite{nvidia_cub_radixsort}.
RT-CDF can therefore be advantageous in workspace size when $R$ is sufficiently small relative to $n$.

This advantage concerns the global-memory capacity required for workspace and does not imply that the number of global-memory accesses is always small.
As discussed above, RT-CDF histogram construction repeatedly examines the input array according to the number of range tiles $m$, increasing input-array traffic for large ranges.
RT-CDF can thus process small ranges quickly with limited workspace, whereas repeated input scans become a performance limitation as the range grows.

The tile width for the main evaluation was fixed for each range on the basis of the preliminary evaluation with uniform and normal inputs and was not separately tuned for all-equal inputs.
The all-equal results therefore represent RT-CDF performance under a distribution-independent configuration rather than under a tile width specialized for that input.
Neither Figure~\ref{fig:rtcdf-kernel-breakdown} nor the main-evaluation results show a substantial degradation for all-equal inputs under this configuration, indicating that the implementation is not strongly dependent on the input distribution.

The evaluation used a single NVIDIA GeForce RTX 4090 GPU.
On GPUs with different numbers of SMs, shared-memory capacities, memory bandwidths, and atomic-operation performance, the optimal tile width, the time ratios among processing stages, and the relative performance against the baselines may change.
In particular, because tile width directly affects occupancy through the shared-memory use per thread block, the range size at which performance begins to degrade is likely to differ on GPUs with different shared-memory limits.
Evaluation on other GPU architectures remains future work.

%% file: tables/table_environment.tex
\begin{table}[H]
\centering
\small
\caption{Experimental environment}
\label{tab:environment}
\begin{tabular}{ll}
\toprule
Item & Setting \\
\midrule
GPU & NVIDIA GeForce RTX 4090 \\
Number of SMs & 128 \\
Compute Capability & 8.9 \\
CUDA Toolkit & 13.0.1 \\
Default shared-memory limit per thread block & 48 KiB \\
Opt-in shared-memory limit per thread block & 99 KiB \\
Shared-memory limit per SM & 100 KiB \\
\bottomrule
\end{tabular}
\end{table}

%% file: tables/table_preliminary_measurement.tex
\begin{table}[H]
\centering
\small
\caption{Preliminary-evaluation conditions for selecting the implementation}
\label{tab:preliminary-measurement}
\begin{tabular}{lp{0.67\linewidth}}
\toprule
Item & Setting \\
\midrule
Input
  & 32-bit integer array satisfying $0\le A_i<R$ \\

Range size
  & $R=2^7,2^8,\ldots,2^{18}$ \\

Input size
  & $10^6,\ 2\times10^6,\ 5\times10^6,\ 10^7,\
     2\times10^7,\ 5\times10^7,\ 10^8,\
     2\times10^8,\ 5\times10^8,\ 10^9$ \\

Input distribution
  & Uniform distribution, Normal distribution
    (mean $(R-1)/2$ and standard deviation $0.125R$) \\

Candidate tile widths
  & $128,\ 256,\ 512,\ 1024,\ 2048,\ 4096,\ 6144,\
     8192,\ 12288,\ 16384,\ 20480,\ 24576$ \\

Candidate constraints
  & $T\le R$, $\lceil R/T\rceil\le16$, 
    and shared-memory use not exceeding the per-block limit \\

Measurements
  & Three independent runs per condition; 10 measurements after warm-up in each run \\

Representative value
  & Median of 10 measurements per run, followed by the median of the three run medians \\

Measured time
  & Host-side wall-clock time including allocation and deallocation of method-specific workspace \\
\bottomrule
\end{tabular}
\end{table}

%% file: figures/fig_output_method_heatmap.tex
\begin{figure}[H]
\centering
\pgfplotsset{
  colormap={outputmethodratio}{
    rgb255(0cm)=(49,54,149);
    rgb255(0.5cm)=(116,173,209);
    rgb255(0.9cm)=(247,247,247);
    rgb255(1.3cm)=(244,109,67);
    rgb255(1.8cm)=(165,0,38)
  },
  outputmethodlabels/.style={
    only marks,
    mark=none,
    nodes near coords={\pgfmathprintnumber[fixed,precision=2,fixed zerofill]{\pgfplotspointmeta}},
    nodes near coords align={center},
    every node near coord/.append style={
      anchor=center,
      inner sep=0pt,
      font=\fontsize{5.2}{5.2}\selectfont
    }
  },
  outputmethodaxis/.style={
    width=\linewidth,
    height=0.94\linewidth,
    axis equal image,
    enlargelimits=false,
    xmin=-0.5, xmax=9.5,
    ymin=-0.5, ymax=11.5,
    xtick={0,1,2,3,4,5,6,7,8,9},
    xticklabels={
      {$10^6$},
      {$2\times10^6$},
      {$5\times10^6$},
      {$10^7$},
      {$2\times10^7$},
      {$5\times10^7$},
      {$10^8$},
      {$2\times10^8$},
      {$5\times10^8$},
      {$10^9$}
    },
    ytick={0,1,2,3,4,5,6,7,8,9,10,11},
    yticklabels={$2^7$,$2^8$,$2^9$,$2^{10}$,$2^{11}$,$2^{12}$,$2^{13}$,$2^{14}$,$2^{15}$,$2^{16}$,$2^{17}$,$2^{18}$},
    xlabel={Input size $n$},
    ylabel={Range size $R$},
    xticklabel style={rotate=45,anchor=east,font=\scriptsize},
    yticklabel style={font=\scriptsize},
    label style={font=\small},
    colormap name=outputmethodratio,
    point meta min=-1.35,
    point meta max=1.35,
    xmajorgrids=false, ymajorgrids=false,
    xminorgrids=true, yminorgrids=true,
    minor x tick num=1, minor y tick num=1,
    minor tick length=0pt,
    minor grid style={draw=black!32,line width=0.2pt},
    colorbar,
    colorbar style={
      ylabel={$\rho$},
      ytick={-1.321928,-0.736966,-0.321928,0,0.321928,0.584963,0.847997},
      yticklabels={$0.4$,$0.6$,$0.8$,$1.0$,$1.25$,$1.5$,$1.8$},
      tick label style={font=\scriptsize},
      label style={font=\small}
    }
  }
}

\begin{subfigure}[t]{0.64\textwidth}
\centering
\begin{tikzpicture}
\begin{axis}[outputmethodaxis]
\addplot[matrix plot*,mesh/cols=10,point meta=explicit]
table[x=x,y=y,meta=log2ratio,col sep=comma,row sep=\\] {
x,y,n,R,ratio,log2ratio\\
0,0,1000000,128,1.330939772,0.412445288\\
1,0,2000000,128,1.325780956,0.406842435\\
2,0,5000000,128,1.339247599,0.421422709\\
3,0,10000000,128,1.190224790,0.251234071\\
4,0,20000000,128,1.460472085,0.546434783\\
5,0,50000000,128,1.784268845,0.835333010\\
6,0,100000000,128,1.801457027,0.849164238\\
7,0,200000000,128,1.788292701,0.838582890\\
8,0,500000000,128,1.793112175,0.842465745\\
9,0,1000000000,128,1.770928464,0.824505936\\
0,1,1000000,256,1.314162441,0.394143615\\
1,1,2000000,256,1.346945478,0.429691455\\
2,1,5000000,256,1.274243423,0.349640906\\
3,1,10000000,256,1.150535058,0.202304945\\
4,1,20000000,256,1.435080433,0.521131599\\
5,1,50000000,256,1.762958846,0.817998797\\
6,1,100000000,256,1.784844892,0.835798705\\
7,1,200000000,256,1.809775596,0.855810820\\
8,1,500000000,256,1.784782409,0.835748199\\
9,1,1000000000,256,1.767135380,0.821412569\\
0,2,1000000,512,1.166406737,0.222070958\\
1,2,2000000,512,1.229840976,0.298471781\\
2,2,5000000,512,1.214646057,0.280535980\\
3,2,10000000,512,1.115138406,0.157222783\\
4,2,20000000,512,1.393380665,0.478589450\\
5,2,50000000,512,1.743375144,0.801883046\\
6,2,100000000,512,1.768332629,0.822389676\\
7,2,200000000,512,1.800085492,0.848065427\\
8,2,500000000,512,1.778943164,0.831020418\\
9,2,1000000000,512,1.750095896,0.807433976\\
0,3,1000000,1024,1.105249528,0.144372119\\
1,3,2000000,1024,1.130286825,0.176688922\\
2,3,5000000,1024,1.156233573,0.209432869\\
3,3,10000000,1024,1.057720206,0.080958049\\
4,3,20000000,1024,1.375082622,0.459518306\\
5,3,50000000,1024,1.725734436,0.787210473\\
6,3,100000000,1024,1.768363138,0.822414566\\
7,3,200000000,1024,1.782523035,0.833920720\\
8,3,500000000,1024,1.777776341,0.830073833\\
9,3,1000000000,1024,1.746604770,0.804553186\\
0,4,1000000,2048,1.176608026,0.234633783\\
1,4,2000000,2048,1.084259077,0.116709521\\
2,4,5000000,2048,1.182429267,0.241753884\\
3,4,10000000,2048,1.033533786,0.047585550\\
4,4,20000000,2048,1.338547342,0.420668165\\
5,4,50000000,2048,1.715614911,0.778725760\\
6,4,100000000,2048,1.773731777,0.826787863\\
7,4,200000000,2048,1.778451823,0.830621893\\
8,4,500000000,2048,1.758138388,0.814048633\\
9,4,1000000000,2048,1.731453796,0.791983890\\
0,5,1000000,4096,0.995662860,-0.006270779\\
1,5,2000000,4096,1.003530450,0.005084393\\
2,5,5000000,4096,1.039252672,0.055546458\\
3,5,10000000,4096,0.991389107,-0.012476688\\
4,5,20000000,4096,1.282967118,0.359484195\\
5,5,50000000,4096,1.722402802,0.784422571\\
6,5,100000000,4096,1.765803472,0.820324785\\
7,5,200000000,4096,1.761440649,0.816755865\\
8,5,500000000,4096,1.754416870,0.810991590\\
9,5,1000000000,4096,1.720314186,0.782672073\\
0,6,1000000,8192,0.795691984,-0.329718030\\
1,6,2000000,8192,0.835057316,-0.260052871\\
2,6,5000000,8192,0.850505684,-0.233607216\\
3,6,10000000,8192,0.910438291,-0.135366859\\
4,6,20000000,8192,1.196172952,0.258426001\\
5,6,50000000,8192,1.674146560,0.743425831\\
6,6,100000000,8192,1.763685246,0.818593115\\
7,6,200000000,8192,1.753944767,0.810603317\\
8,6,500000000,8192,1.747255091,0.805090250\\
9,6,1000000000,8192,1.703165390,0.768218538\\
0,7,1000000,16384,0.662831655,-0.593285591\\
1,7,2000000,16384,0.771812197,-0.373678253\\
2,7,5000000,16384,0.773881226,-0.369815934\\
3,7,10000000,16384,0.812782574,-0.299058623\\
4,7,20000000,16384,1.166127222,0.221725192\\
5,7,50000000,16384,1.602274940,0.680121726\\
6,7,100000000,16384,1.759572767,0.815225178\\
7,7,200000000,16384,1.754468680,0.811034194\\
8,7,500000000,16384,1.753507045,0.810243227\\
9,7,1000000000,16384,1.698223491,0.764026334\\
0,8,1000000,32768,0.617337023,-0.695869780\\
1,8,2000000,32768,0.627875208,-0.671450247\\
2,8,5000000,32768,0.683723499,-0.548515084\\
3,8,10000000,32768,0.774874763,-0.367964938\\
4,8,20000000,32768,1.074035041,0.103041062\\
5,8,50000000,32768,1.496425844,0.581520787\\
6,8,100000000,32768,1.587886978,0.667108228\\
7,8,200000000,32768,1.691807513,0.758565433\\
8,8,500000000,32768,1.665737597,0.736161152\\
9,8,1000000000,32768,1.631884434,0.706538893\\
0,9,1000000,65536,0.481333958,-1.054889886\\
1,9,2000000,65536,0.548790527,-0.865672518\\
2,9,5000000,65536,0.613266211,-0.705414631\\
3,9,10000000,65536,0.775397233,-0.366992509\\
4,9,20000000,65536,1.021022294,0.030014367\\
5,9,50000000,65536,1.314567276,0.394587977\\
6,9,100000000,65536,1.435385388,0.521438139\\
7,9,200000000,65536,1.494029309,0.579208451\\
8,9,500000000,65536,1.487479801,0.572870078\\
9,9,1000000000,65536,1.453781224,0.539810178\\
0,10,1000000,131072,0.398149009,-1.328619627\\
1,10,2000000,131072,0.460574220,-1.118494433\\
2,10,5000000,131072,0.602832754,-0.730170289\\
3,10,10000000,131072,0.767820956,-0.381158159\\
4,10,20000000,131072,0.949492322,-0.074771761\\
5,10,50000000,131072,1.168594854,0.224774842\\
6,10,100000000,131072,1.234790760,0.304266592\\
7,10,200000000,131072,1.294869845,0.372807091\\
8,10,500000000,131072,1.306479646,0.385684649\\
9,10,1000000000,131072,1.298409521,0.376745484\\
0,11,1000000,262144,0.491479472,-1.024796936\\
1,11,2000000,262144,0.541749430,-0.884302366\\
2,11,5000000,262144,0.620591124,-0.688285033\\
3,11,10000000,262144,0.784863158,-0.349486956\\
4,11,20000000,262144,0.910639606,-0.135047888\\
5,11,50000000,262144,1.043425543,0.061327656\\
6,11,100000000,262144,1.076570012,0.106442145\\
7,11,200000000,262144,1.095637382,0.131770395\\
8,11,500000000,262144,1.108241319,0.148272062\\
9,11,1000000000,262144,1.112553688,0.153874957\\
};
\addplot[outputmethodlabels,text=black,point meta=explicit]
table[x=x,y=y,meta=ratio,col sep=comma,row sep=\\] {
x,y,ratio\\
3,0,1.190224790\\
0,1,1.314162441\\
2,1,1.274243423\\
3,1,1.150535058\\
0,2,1.166406737\\
1,2,1.229840976\\
2,2,1.214646057\\
3,2,1.115138406\\
0,3,1.105249528\\
1,3,1.130286825\\
2,3,1.156233573\\
3,3,1.057720206\\
0,4,1.176608026\\
1,4,1.084259077\\
2,4,1.182429267\\
3,4,1.033533786\\
0,5,0.995662860\\
1,5,1.003530450\\
2,5,1.039252672\\
3,5,0.991389107\\
4,5,1.282967118\\
0,6,0.795691984\\
1,6,0.835057316\\
2,6,0.850505684\\
3,6,0.910438291\\
4,6,1.196172952\\
1,7,0.771812197\\
2,7,0.773881226\\
3,7,0.812782574\\
4,7,1.166127222\\
3,8,0.774874763\\
4,8,1.074035041\\
3,9,0.775397233\\
4,9,1.021022294\\
5,9,1.314567276\\
3,10,0.767820956\\
4,10,0.949492322\\
5,10,1.168594854\\
6,10,1.234790760\\
7,10,1.294869845\\
8,10,1.306479646\\
9,10,1.298409521\\
3,11,0.784863158\\
4,11,0.910639606\\
5,11,1.043425543\\
6,11,1.076570012\\
7,11,1.095637382\\
8,11,1.108241319\\
9,11,1.112553688\\
};
\addplot[outputmethodlabels,text=white,point meta=explicit]
table[x=x,y=y,meta=ratio,col sep=comma,row sep=\\] {
x,y,ratio\\
0,0,1.330939772\\
1,0,1.325780956\\
2,0,1.339247599\\
4,0,1.460472085\\
5,0,1.784268845\\
6,0,1.801457027\\
7,0,1.788292701\\
8,0,1.793112175\\
9,0,1.770928464\\
1,1,1.346945478\\
4,1,1.435080433\\
5,1,1.762958846\\
6,1,1.784844892\\
7,1,1.809775596\\
8,1,1.784782409\\
9,1,1.767135380\\
4,2,1.393380665\\
5,2,1.743375144\\
6,2,1.768332629\\
7,2,1.800085492\\
8,2,1.778943164\\
9,2,1.750095896\\
4,3,1.375082622\\
5,3,1.725734436\\
6,3,1.768363138\\
7,3,1.782523035\\
8,3,1.777776341\\
9,3,1.746604770\\
4,4,1.338547342\\
5,4,1.715614911\\
6,4,1.773731777\\
7,4,1.778451823\\
8,4,1.758138388\\
9,4,1.731453796\\
5,5,1.722402802\\
6,5,1.765803472\\
7,5,1.761440649\\
8,5,1.754416870\\
9,5,1.720314186\\
5,6,1.674146560\\
6,6,1.763685246\\
7,6,1.753944767\\
8,6,1.747255091\\
9,6,1.703165390\\
0,7,0.662831655\\
5,7,1.602274940\\
6,7,1.759572767\\
7,7,1.754468680\\
8,7,1.753507045\\
9,7,1.698223491\\
0,8,0.617337023\\
1,8,0.627875208\\
2,8,0.683723499\\
5,8,1.496425844\\
6,8,1.587886978\\
7,8,1.691807513\\
8,8,1.665737597\\
9,8,1.631884434\\
0,9,0.481333958\\
1,9,0.548790527\\
2,9,0.613266211\\
6,9,1.435385388\\
7,9,1.494029309\\
8,9,1.487479801\\
9,9,1.453781224\\
0,10,0.398149009\\
1,10,0.460574220\\
2,10,0.602832754\\
0,11,0.491479472\\
1,11,0.541749430\\
2,11,0.620591124\\
};
\end{axis}
\end{tikzpicture}
\caption{Uniform distribution}
\label{fig:output-method-heatmap-uniform}
\end{subfigure}

\vspace{1mm}

\begin{subfigure}[t]{0.64\textwidth}
\centering
\begin{tikzpicture}
\begin{axis}[outputmethodaxis]
\addplot[matrix plot*,mesh/cols=10,point meta=explicit]
table[x=x,y=y,meta=log2ratio,col sep=comma,row sep=\\] {
x,y,n,R,ratio,log2ratio\\
0,0,1000000,128,1.308830042,0.388277768\\
1,0,2000000,128,1.322502827,0.403270806\\
2,0,5000000,128,1.314056654,0.394027477\\
3,0,10000000,128,1.185388808,0.245360341\\
4,0,20000000,128,1.460474207,0.546436880\\
5,0,50000000,128,1.779058874,0.831114254\\
6,0,100000000,128,1.801657943,0.849325131\\
7,0,200000000,128,1.788858924,0.839039616\\
8,0,500000000,128,1.794803821,0.843826161\\
9,0,1000000000,128,1.766061845,0.820535865\\
0,1,1000000,256,1.325268702,0.406284900\\
1,1,2000000,256,1.340841278,0.423138468\\
2,1,5000000,256,1.240788647,0.311257392\\
3,1,10000000,256,1.138845341,0.187571838\\
4,1,20000000,256,1.409500839,0.495184337\\
5,1,50000000,256,1.755629779,0.811988647\\
6,1,100000000,256,1.785861482,0.836620184\\
7,1,200000000,256,1.804761285,0.851808025\\
8,1,500000000,256,1.786845045,0.837414529\\
9,1,1000000000,256,1.759833777,0.815439167\\
0,2,1000000,512,1.209253743,0.274117003\\
1,2,2000000,512,1.213314904,0.278954037\\
2,2,5000000,512,1.185282915,0.245231457\\
3,2,10000000,512,1.099769382,0.137201027\\
4,2,20000000,512,1.397741078,0.483097136\\
5,2,50000000,512,1.742811963,0.801416921\\
6,2,100000000,512,1.771208753,0.824734257\\
7,2,200000000,512,1.800461228,0.848366532\\
8,2,500000000,512,1.783222577,0.834486787\\
9,2,1000000000,512,1.757406074,0.813447585\\
0,3,1000000,1024,1.119322801,0.162626154\\
1,3,2000000,1024,1.218178974,0.284726108\\
2,3,5000000,1024,1.097355164,0.134030536\\
3,3,10000000,1024,1.086416667,0.119577517\\
4,3,20000000,1024,1.366829715,0.450833517\\
5,3,50000000,1024,1.721684989,0.783821201\\
6,3,100000000,1024,1.774857928,0.827703546\\
7,3,200000000,1024,1.792811593,0.842223883\\
8,3,500000000,1024,1.769019181,0.822949691\\
9,3,1000000000,1024,1.747052169,0.804922690\\
0,4,1000000,2048,1.120483506,0.164121412\\
1,4,2000000,2048,1.183237835,0.242740090\\
2,4,5000000,2048,1.070148418,0.097810896\\
3,4,10000000,2048,1.050585317,0.071193327\\
4,4,20000000,2048,1.330030385,0.411459205\\
5,4,50000000,2048,1.715929488,0.778990270\\
6,4,100000000,2048,1.757382993,0.813428638\\
7,4,200000000,2048,1.771406485,0.824895305\\
8,4,500000000,2048,1.768860644,0.822820392\\
9,4,1000000000,2048,1.729952427,0.790732365\\
0,5,1000000,4096,0.979154637,-0.030391374\\
1,5,2000000,4096,1.015429005,0.022089375\\
2,5,5000000,4096,1.014289282,0.020469177\\
3,5,10000000,4096,0.968062436,-0.046827996\\
4,5,20000000,4096,1.279939506,0.356075625\\
5,5,50000000,4096,1.714006284,0.777372399\\
6,5,100000000,4096,1.757502641,0.813526857\\
7,5,200000000,4096,1.766536633,0.820923667\\
8,5,500000000,4096,1.752280703,0.809233903\\
9,5,1000000000,4096,1.718547932,0.781190090\\
0,6,1000000,8192,0.768084286,-0.380663460\\
1,6,2000000,8192,0.829135926,-0.270319462\\
2,6,5000000,8192,0.854999015,-0.226005338\\
3,6,10000000,8192,0.879534746,-0.185187522\\
4,6,20000000,8192,1.190549419,0.251627508\\
5,6,50000000,8192,1.670934143,0.740654873\\
6,6,100000000,8192,1.760405614,0.815907878\\
7,6,200000000,8192,1.748940600,0.806481291\\
8,6,500000000,8192,1.742549370,0.801199531\\
9,6,1000000000,8192,1.713740054,0.777148293\\
0,7,1000000,16384,0.729057946,-0.455894610\\
1,7,2000000,16384,0.757096774,-0.401450374\\
2,7,5000000,16384,0.785605112,-0.348123778\\
3,7,10000000,16384,0.839627339,-0.252178952\\
4,7,20000000,16384,1.160492838,0.214737619\\
5,7,50000000,16384,1.600834243,0.678823933\\
6,7,100000000,16384,1.730333248,0.791049916\\
7,7,200000000,16384,1.758060530,0.813984744\\
8,7,500000000,16384,1.762066639,0.817268486\\
9,7,1000000000,16384,1.690454357,0.757411064\\
0,8,1000000,32768,0.596490700,-0.745428451\\
1,8,2000000,32768,0.639724861,-0.644476544\\
2,8,5000000,32768,0.680635991,-0.555044653\\
3,8,10000000,32768,0.773366467,-0.370775885\\
4,8,20000000,32768,1.067716134,0.094528139\\
5,8,50000000,32768,1.493224953,0.578431522\\
6,8,100000000,32768,1.651164002,0.723483423\\
7,8,200000000,32768,1.645730260,0.718727893\\
8,8,500000000,32768,1.644307743,0.717480334\\
9,8,1000000000,32768,1.583706264,0.663304778\\
0,9,1000000,65536,0.487032408,-1.037910320\\
1,9,2000000,65536,0.559783890,-0.837058126\\
2,9,5000000,65536,0.623539472,-0.681447205\\
3,9,10000000,65536,0.765358469,-0.385792477\\
4,9,20000000,65536,1.002407488,0.003469097\\
5,9,50000000,65536,1.309422355,0.388930514\\
6,9,100000000,65536,1.440054416,0.526123329\\
7,9,200000000,65536,1.482307938,0.567845188\\
8,9,500000000,65536,1.466432942,0.552311100\\
9,9,1000000000,65536,1.433563005,0.519605312\\
0,10,1000000,131072,0.403485244,-1.309412182\\
1,10,2000000,131072,0.456886159,-1.130093357\\
2,10,5000000,131072,0.593269724,-0.753239934\\
3,10,10000000,131072,0.762370037,-0.391436676\\
4,10,20000000,131072,0.939528026,-0.089991898\\
5,10,50000000,131072,1.176487319,0.234485771\\
6,10,100000000,131072,1.250384604,0.322371919\\
7,10,200000000,131072,1.294476888,0.372369207\\
8,10,500000000,131072,1.293231059,0.370980062\\
9,10,1000000000,131072,1.281866988,0.358246570\\
0,11,1000000,262144,0.482517029,-1.051348237\\
1,11,2000000,262144,0.521697441,-0.938714738\\
2,11,5000000,262144,0.616265458,-0.698376166\\
3,11,10000000,262144,0.766656209,-0.383348319\\
4,11,20000000,262144,0.893474956,-0.162500803\\
5,11,50000000,262144,1.040131295,0.056765650\\
6,11,100000000,262144,1.075245239,0.104665743\\
7,11,200000000,262144,1.095876275,0.132084926\\
8,11,500000000,262144,1.108636234,0.148786066\\
9,11,1000000000,262144,1.113644156,0.155288320\\
};
\addplot[outputmethodlabels,text=black,point meta=explicit]
table[x=x,y=y,meta=ratio,col sep=comma,row sep=\\] {
x,y,ratio\\
0,0,1.308830042\\
2,0,1.314056654\\
3,0,1.185388808\\
2,1,1.240788647\\
3,1,1.138845341\\
0,2,1.209253743\\
1,2,1.213314904\\
2,2,1.185282915\\
3,2,1.099769382\\
0,3,1.119322801\\
1,3,1.218178974\\
2,3,1.097355164\\
3,3,1.086416667\\
0,4,1.120483506\\
1,4,1.183237835\\
2,4,1.070148418\\
3,4,1.050585317\\
0,5,0.979154637\\
1,5,1.015429005\\
2,5,1.014289282\\
3,5,0.968062436\\
4,5,1.279939506\\
0,6,0.768084286\\
1,6,0.829135926\\
2,6,0.854999015\\
3,6,0.879534746\\
4,6,1.190549419\\
2,7,0.785605112\\
3,7,0.839627339\\
4,7,1.160492838\\
3,8,0.773366467\\
4,8,1.067716134\\
3,9,0.765358469\\
4,9,1.002407488\\
5,9,1.309422355\\
3,10,0.762370037\\
4,10,0.939528026\\
5,10,1.176487319\\
6,10,1.250384604\\
7,10,1.294476888\\
8,10,1.293231059\\
9,10,1.281866988\\
3,11,0.766656209\\
4,11,0.893474956\\
5,11,1.040131295\\
6,11,1.075245239\\
7,11,1.095876275\\
8,11,1.108636234\\
9,11,1.113644156\\
};
\addplot[outputmethodlabels,text=white,point meta=explicit]
table[x=x,y=y,meta=ratio,col sep=comma,row sep=\\] {
x,y,ratio\\
1,0,1.322502827\\
4,0,1.460474207\\
5,0,1.779058874\\
6,0,1.801657943\\
7,0,1.788858924\\
8,0,1.794803821\\
9,0,1.766061845\\
0,1,1.325268702\\
1,1,1.340841278\\
4,1,1.409500839\\
5,1,1.755629779\\
6,1,1.785861482\\
7,1,1.804761285\\
8,1,1.786845045\\
9,1,1.759833777\\
4,2,1.397741078\\
5,2,1.742811963\\
6,2,1.771208753\\
7,2,1.800461228\\
8,2,1.783222577\\
9,2,1.757406074\\
4,3,1.366829715\\
5,3,1.721684989\\
6,3,1.774857928\\
7,3,1.792811593\\
8,3,1.769019181\\
9,3,1.747052169\\
4,4,1.330030385\\
5,4,1.715929488\\
6,4,1.757382993\\
7,4,1.771406485\\
8,4,1.768860644\\
9,4,1.729952427\\
5,5,1.714006284\\
6,5,1.757502641\\
7,5,1.766536633\\
8,5,1.752280703\\
9,5,1.718547932\\
5,6,1.670934143\\
6,6,1.760405614\\
7,6,1.748940600\\
8,6,1.742549370\\
9,6,1.713740054\\
0,7,0.729057946\\
1,7,0.757096774\\
5,7,1.600834243\\
6,7,1.730333248\\
7,7,1.758060530\\
8,7,1.762066639\\
9,7,1.690454357\\
0,8,0.596490700\\
1,8,0.639724861\\
2,8,0.680635991\\
5,8,1.493224953\\
6,8,1.651164002\\
7,8,1.645730260\\
8,8,1.644307743\\
9,8,1.583706264\\
0,9,0.487032408\\
1,9,0.559783890\\
2,9,0.623539472\\
6,9,1.440054416\\
7,9,1.482307938\\
8,9,1.466432942\\
9,9,1.433563005\\
0,10,0.403485244\\
1,10,0.456886159\\
2,10,0.593269724\\
0,11,0.482517029\\
1,11,0.521697441\\
2,11,0.616265458\\
};
\end{axis}
\end{tikzpicture}
\caption{Normal distribution}
\label{fig:output-method-heatmap-gaussian}
\end{subfigure}

\caption{Execution-time ratio $\rho$ of the prefix-maximum and binary-search methods}
\label{fig:output-method-heatmap}
\end{figure}
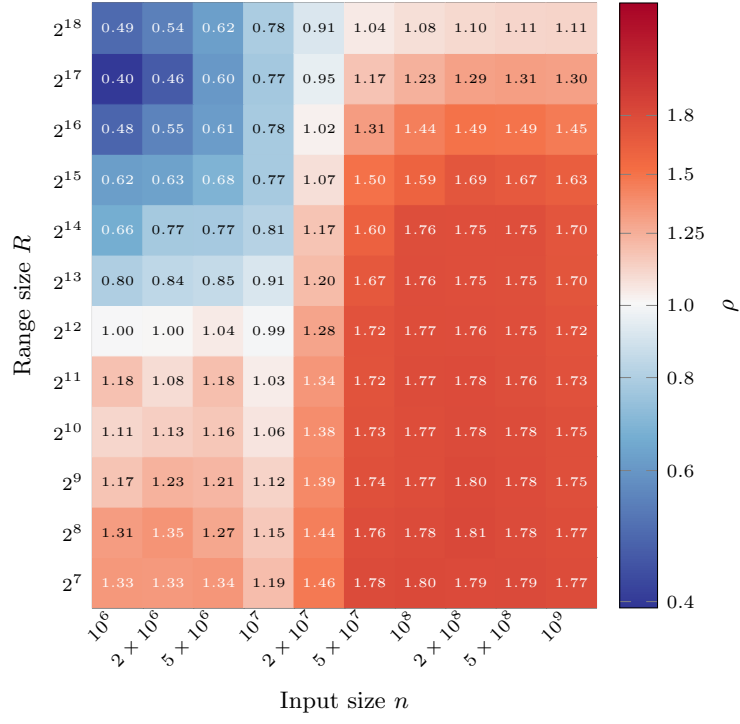
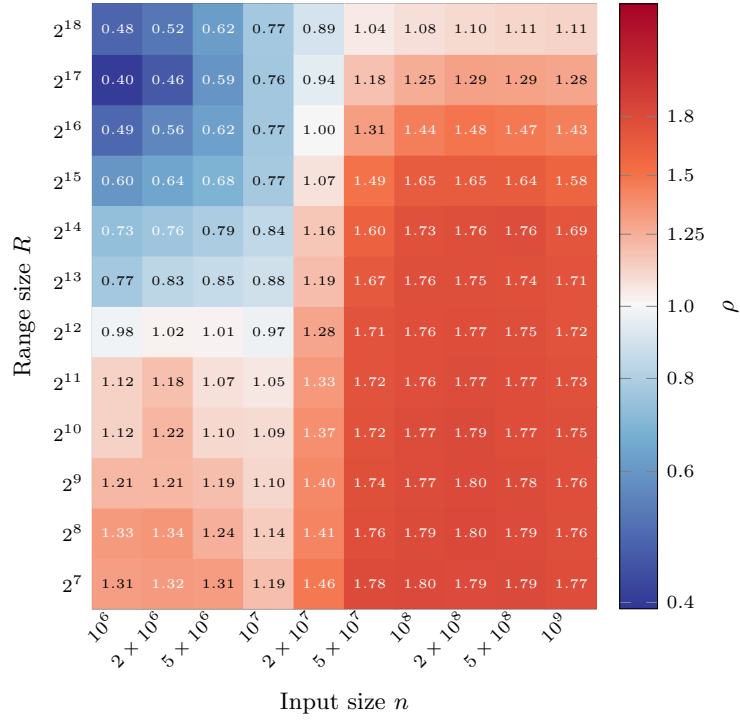

%% file: tables/table_tile_width_relative.tex
\begin{table}[H]
\centering
\caption{Geometric mean of normalized execution time for each tile width}
\label{tab:tile-width-relative-time}
\scriptsize
\setlength{\tabcolsep}{2.2pt}
\renewcommand{\arraystretch}{1.18}
\newcommand{\nottested}{\cellcolor[gray]{0.93}{\color[gray]{0.45}--}}
\begin{tabular}{@{}c*{12}{c}@{}}
\toprule
& \multicolumn{12}{c}{Tile width $T$} \\
\cmidrule(l){2-13}
Range size $R$
& 128 & 256 & 512 & 1{,}024 & 2{,}048 & 4{,}096
& 6{,}144 & 8{,}192 & 12{,}288 & 16{,}384 & 20{,}480 & 24{,}576 \\
\midrule
$2^{7}$  & 1.000 & \nottested & \nottested & \nottested & \nottested & \nottested & \nottested & \nottested & \nottested & \nottested & \nottested & \nottested \\
\addlinespace[1.5pt]
$2^{8}$  & 1.134 & 1.000 & \nottested & \nottested & \nottested & \nottested & \nottested & \nottested & \nottested & \nottested & \nottested & \nottested \\
\addlinespace[1.5pt]
$2^{9}$  & 1.306 & 1.102 & 1.000 & \nottested & \nottested & \nottested & \nottested & \nottested & \nottested & \nottested & \nottested & \nottested \\
\addlinespace[1.5pt]
$2^{10}$ & 1.806 & 1.274 & 1.089 & 1.001 & \nottested & \nottested & \nottested & \nottested & \nottested & \nottested & \nottested & \nottested \\
\addlinespace[1.5pt]
$2^{11}$ & 2.758 & 1.759 & 1.268 & 1.078 & 1.001 & \nottested & \nottested & \nottested & \nottested & \nottested & \nottested & \nottested \\
\addlinespace[1.5pt]
$2^{12}$ & \nottested & 2.627 & 1.688 & 1.213 & 1.054 & 1.001 & \nottested & \nottested & \nottested & \nottested & \nottested & \nottested \\
\addlinespace[1.5pt]
$2^{13}$ & \nottested & \nottested & 2.420 & 1.555 & 1.129 & 1.013 & 1.030 & 1.017 & \nottested & \nottested & \nottested & \nottested \\
\addlinespace[1.5pt]
$2^{14}$ & \nottested & \nottested & \nottested & 2.310 & 1.480 & 1.112 & 1.052 & 1.029 & 1.119 & 1.216 & \nottested & \nottested \\
\addlinespace[1.5pt]
$2^{15}$ & \nottested & \nottested & \nottested & \nottested & 1.951 & 1.278 & 1.127 & 1.010 & 1.057 & 1.297 & 1.323 & 1.376 \\
\addlinespace[1.5pt]
$2^{16}$ & \nottested & \nottested & \nottested & \nottested & \nottested & 1.450 & 1.171 & 1.004 & 1.063 & 1.286 & 1.318 & 1.199 \\
\addlinespace[1.5pt]
$2^{17}$ & \nottested & \nottested & \nottested & \nottested & \nottested & \nottested & \nottested & 1.013 & 1.012 & 1.285 & 1.219 & 1.136 \\
\addlinespace[1.5pt]
$2^{18}$ & \nottested & \nottested & \nottested & \nottested & \nottested & \nottested & \nottested & \nottested & \nottested & 1.202 & 1.123 & 1.001 \\
\bottomrule
\end{tabular}
\end{table}

%% file: tables/table_selected_tile_width.tex
\begin{table}[H]
\centering
\caption{Tile widths and numbers of tiles used in the main evaluation}
\label{tab:selected-tile-width}
\small
\setlength{\tabcolsep}{14pt}
\renewcommand{\arraystretch}{1.12}
\begin{tabular}{ccc}
\toprule
Range size $R$ & Tile width $T$ & Number of tiles $m=\lceil R/T\rceil$ \\
\midrule
$2^{7}$  & 128      & 1  \\
$2^{8}$  & 256      & 1  \\
$2^{9}$  & 512      & 1  \\
$2^{10}$ & 1{,}024  & 1  \\
$2^{11}$ & 2{,}048  & 1  \\
$2^{12}$ & 4{,}096  & 1  \\
$2^{13}$ & 4{,}096  & 2  \\
$2^{14}$ & 8{,}192  & 2  \\
$2^{15}$ & 8{,}192  & 4  \\
$2^{16}$ & 8{,}192  & 8  \\
$2^{17}$ & 12{,}288 & 11 \\
$2^{18}$ & 24{,}576 & 11 \\
\bottomrule
\end{tabular}
\end{table}

%% file: tables/table_comparison_targets.tex
\begin{table}[H]
\centering
\small
\caption{Baselines}
\label{tab:comparison-targets}
\begin{tabular}{p{0.15\linewidth}p{0.48\linewidth}p{0.27\linewidth}}
\toprule
Method & Implementation & Role \\
\midrule
CUB & \texttt{DeviceRadixSort} with \texttt{begin\_bit=0} and \texttt{end\_bit=}$\lceil\log_2 R\rceil$ & Stable-sorting baseline using the known range \\
RefHP & Prior implementation refining H-P sort for GPUs & Unstable counting-sort-based baseline \\
Kolonias et al. & Comparison implementation developed in this study from the algorithm of Kolonias et al. & Baseline with a closely related output-generation method \\
RT-CDF & Implementation that partitions the range and generates output from each tile histogram and CDF & Proposed method \\
\bottomrule
\end{tabular}
\end{table}

%% file: tables/table_measurement_method.tex
\begin{table}[H]
\centering
\small
\caption{Measurement conditions for the main evaluation}
\label{tab:measurement-method}
\begin{tabular}{lp{0.72\linewidth}}
\toprule
Item & Setting \\
\midrule
Input & 32-bit integer array satisfying $0 \le A_i < R$ \\
Range size & $R=2^7,2^8,\ldots,2^{18}$ \\
Input size & $10^6,2\times10^6,\ldots,9\times10^6$, $10^7,2\times10^7,\ldots,9\times10^7$, $10^8,2\times10^8,\ldots,10^9$ \\
Input distribution & Uniform, normal (mean $(R-1)/2$, standard deviation $0.125R$), and all-equal \\
Repetitions & Three independent runs per condition; 30 measurements after warm-up in each run \\
Representative value & Median of 30 measurements per run, followed by the median of the three run medians \\
Measured time & Host-side wall-clock time including allocation and deallocation of method-specific workspace \\
Termination condition & For a fixed range and distribution, omit larger input sizes after the execution time exceeds 1000\,ms \\
\bottomrule
\end{tabular}
\end{table}

%% file: tables/table_n1e8_times_uniform.tex
\begin{table}[H]
\centering
\small
\caption{Execution times and speedups for $n=10^8$ (uniform distribution)}
\label{tab:n1e8-times-uniform}
\setlength{\tabcolsep}{3.5pt}
\begin{tabular}{lrrrrcr}
\toprule
$R$ & RT-CDF [ms] & CUB [ms] & RefHP [ms] & Kolonias et al. [ms] & Best baseline & Best baseline/RT-CDF \\
\midrule
$2^{7}$ & \textbf{0.956} & 1.446 & 2.833 & 73.34 & CUB & 1.51 \\
$2^{8}$ & \textbf{0.961} & 1.444 & 2.849 & 73.52 & CUB & 1.50 \\
$2^{9}$ & \textbf{0.968} & 2.561 & 4.088 & 37.00 & CUB & 2.65 \\
$2^{10}$ & \textbf{0.971} & 2.579 & 4.090 & 18.76 & CUB & 2.66 \\
$2^{11}$ & \textbf{0.982} & 2.584 & 4.091 & 9.651 & CUB & 2.63 \\
$2^{12}$ & \textbf{0.993} & 2.579 & 4.089 & 5.085 & CUB & 2.60 \\
$2^{13}$ & \textbf{1.008} & 2.565 & 4.089 & 3.135 & CUB & 2.54 \\
$2^{14}$ & \textbf{1.012} & 2.597 & 4.084 & 3.517 & CUB & 2.57 \\
$2^{15}$ & \textbf{1.187} & 2.610 & 4.078 & 3.982 & CUB & 2.20 \\
$2^{16}$ & \textbf{1.714} & 2.606 & 4.078 & 4.202 & CUB & 1.52 \\
$2^{17}$ & \textbf{2.975} & 3.414 & 3.953 & 5.798 & CUB & 1.15 \\
$2^{18}$ & 5.528 & \textbf{3.421} & 3.994 & 8.642 & CUB & 0.62 \\
\bottomrule
\end{tabular}
\end{table}

%% file: tables/table_n1e8_times_gaussian.tex
\begin{table}[H]
\centering
\small
\caption{Execution times and speedups for $n=10^8$ (normal distribution)}
\label{tab:n1e8-times-gaussian}
\setlength{\tabcolsep}{3.5pt}
\begin{tabular}{lrrrrcr}
\toprule
$R$ & RT-CDF [ms] & CUB [ms] & RefHP [ms] & Kolonias et al. [ms] & Best baseline & Best baseline/RT-CDF \\
\midrule
$2^{7}$ & \textbf{0.956} & 1.435 & 2.953 & 69.75 & CUB & 1.50 \\
$2^{8}$ & \textbf{0.961} & 1.442 & 2.850 & 73.03 & CUB & 1.50 \\
$2^{9}$ & \textbf{0.968} & 2.582 & 4.080 & 36.77 & CUB & 2.67 \\
$2^{10}$ & \textbf{0.973} & 2.585 & 4.078 & 35.31 & CUB & 2.66 \\
$2^{11}$ & \textbf{0.980} & 2.578 & 4.091 & 25.43 & CUB & 2.63 \\
$2^{12}$ & \textbf{0.992} & 2.578 & 4.094 & 14.48 & CUB & 2.60 \\
$2^{13}$ & \textbf{1.007} & 2.580 & 4.096 & 7.738 & CUB & 2.56 \\
$2^{14}$ & \textbf{1.012} & 2.587 & 4.089 & 4.803 & CUB & 2.56 \\
$2^{15}$ & \textbf{1.186} & 2.590 & 4.089 & 4.293 & CUB & 2.18 \\
$2^{16}$ & \textbf{1.707} & 2.607 & 4.078 & 4.274 & CUB & 1.53 \\
$2^{17}$ & \textbf{2.881} & 3.429 & 3.983 & 5.951 & CUB & 1.19 \\
$2^{18}$ & 5.520 & \textbf{3.419} & 3.990 & 8.690 & CUB & 0.62 \\
\bottomrule
\end{tabular}
\end{table}

%% file: tables/table_n1e8_times_all-equal.tex
\begin{table}[H]
\centering
\small
\caption{Execution times and speedups for $n=10^8$ (all-equal input)}
\label{tab:n1e8-times-all-equal}
\setlength{\tabcolsep}{3.5pt}
\begin{tabular}{lrrrrcr}
\toprule
$R$ & RT-CDF [ms] & CUB [ms] & RefHP [ms] & Kolonias et al. [ms] & Best baseline & Best baseline/RT-CDF \\
\midrule
$2^{7}$ & \textbf{0.956} & 1.430 & 1.831 & 267.40 & CUB & 1.50 \\
$2^{8}$ & \textbf{0.961} & 1.430 & 1.852 & 267.41 & CUB & 1.49 \\
$2^{9}$ & \textbf{0.969} & 2.562 & 3.042 & 267.38 & CUB & 2.64 \\
$2^{10}$ & \textbf{0.974} & 2.576 & 3.042 & 267.38 & CUB & 2.65 \\
$2^{11}$ & \textbf{0.981} & 2.576 & 4.093 & 266.65 & CUB & 2.63 \\
$2^{12}$ & \textbf{0.994} & 2.507 & 4.450 & 267.36 & CUB & 2.52 \\
$2^{13}$ & \textbf{1.018} & 2.509 & 6.839 & 267.39 & CUB & 2.46 \\
$2^{14}$ & \textbf{1.029} & 2.509 & 7.056 & 267.89 & CUB & 2.44 \\
$2^{15}$ & \textbf{1.200} & 2.509 & 11.722 & 267.74 & CUB & 2.09 \\
$2^{16}$ & \textbf{1.688} & 2.506 & 11.656 & 267.91 & CUB & 1.48 \\
$2^{17}$ & \textbf{2.767} & 3.387 & 13.259 & 270.19 & CUB & 1.22 \\
$2^{18}$ & 5.333 & \textbf{3.404} & 15.744 & 273.06 & CUB & 0.64 \\
\bottomrule
\end{tabular}
\end{table}

%% file: figures/fig_runtime_r1024_uniform.tex
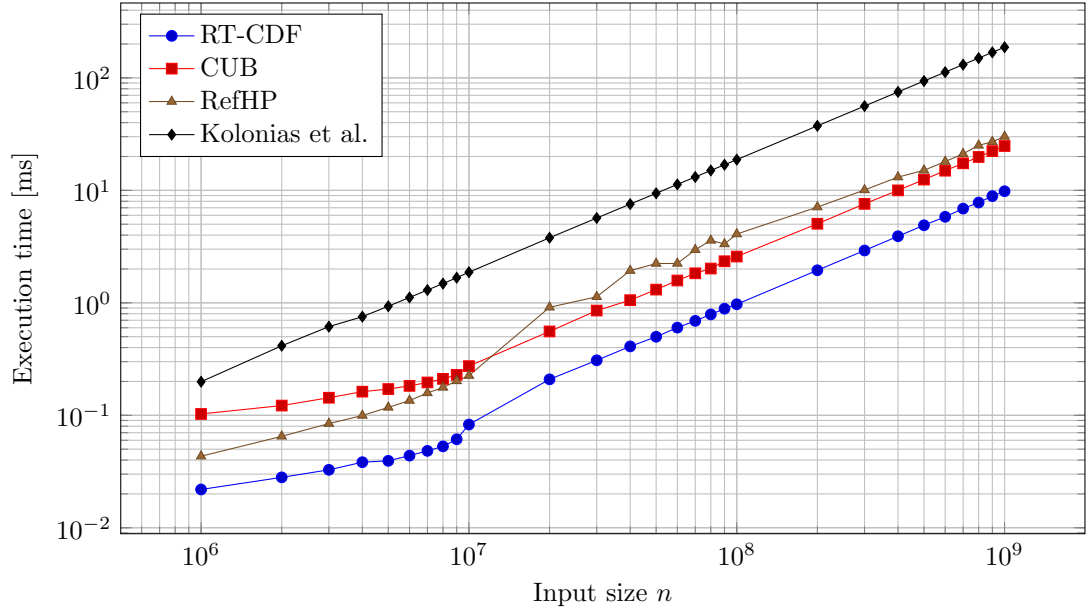
\begin{figure}[H]
\centering
\begin{tikzpicture}
\begin{axis}[
width=0.90\linewidth,height=0.54\linewidth,
xmode=log,ymode=log,log basis x=10,log basis y=10,
xlabel={Input size $n$},ylabel={Execution time [ms]},
grid=both,
unbounded coords=jump,
legend style={at={(0.02,0.98)},anchor=north west,fill=white,draw=black},
legend cell align={left}
]
\addplot+[mark=*] table[x=n,y=rtcdf,col sep=comma] {plotdata/runtime_R1024_uniform_rtx4090.dat};
\addlegendentry{RT-CDF}
\addplot+[mark=square*] table[x=n,y=cub,col sep=comma] {plotdata/runtime_R1024_uniform_rtx4090.dat};
\addlegendentry{CUB}
\addplot+[mark=triangle*] table[x=n,y=refhp,col sep=comma] {plotdata/runtime_R1024_uniform_rtx4090.dat};
\addlegendentry{RefHP}
\addplot+[mark=diamond*] table[x=n,y=kolonias,col sep=comma] {plotdata/runtime_R1024_uniform_rtx4090.dat};
\addlegendentry{Kolonias et al.}
\end{axis}
\end{tikzpicture}
\caption{Execution time for uniformly distributed inputs ($R=2^{10}$)}
\label{fig:runtime-r1024-uniform}
\end{figure}

%% file: figures/fig_runtime_r1024_gaussian.tex
\begin{figure}[H]
\centering
\begin{tikzpicture}
\begin{axis}[
width=0.90\linewidth,height=0.54\linewidth,
xmode=log,ymode=log,log basis x=10,log basis y=10,
xlabel={Input size $n$},ylabel={Execution time [ms]},
grid=both,
unbounded coords=jump,
legend style={at={(0.02,0.98)},anchor=north west,fill=white,draw=black},
legend cell align={left}
]
\addplot+[mark=*] table[x=n,y=rtcdf,col sep=comma] {plotdata/runtime_R1024_gaussian_rtx4090.dat};
\addlegendentry{RT-CDF}
\addplot+[mark=square*] table[x=n,y=cub,col sep=comma] {plotdata/runtime_R1024_gaussian_rtx4090.dat};
\addlegendentry{CUB}
\addplot+[mark=triangle*] table[x=n,y=refhp,col sep=comma] {plotdata/runtime_R1024_gaussian_rtx4090.dat};
\addlegendentry{RefHP}
\addplot+[mark=diamond*] table[x=n,y=kolonias,col sep=comma] {plotdata/runtime_R1024_gaussian_rtx4090.dat};
\addlegendentry{Kolonias et al.}
\end{axis}
\end{tikzpicture}
\caption{Execution time for normally distributed inputs ($R=2^{10}$)}
\label{fig:runtime-r1024-gaussian}
\end{figure}
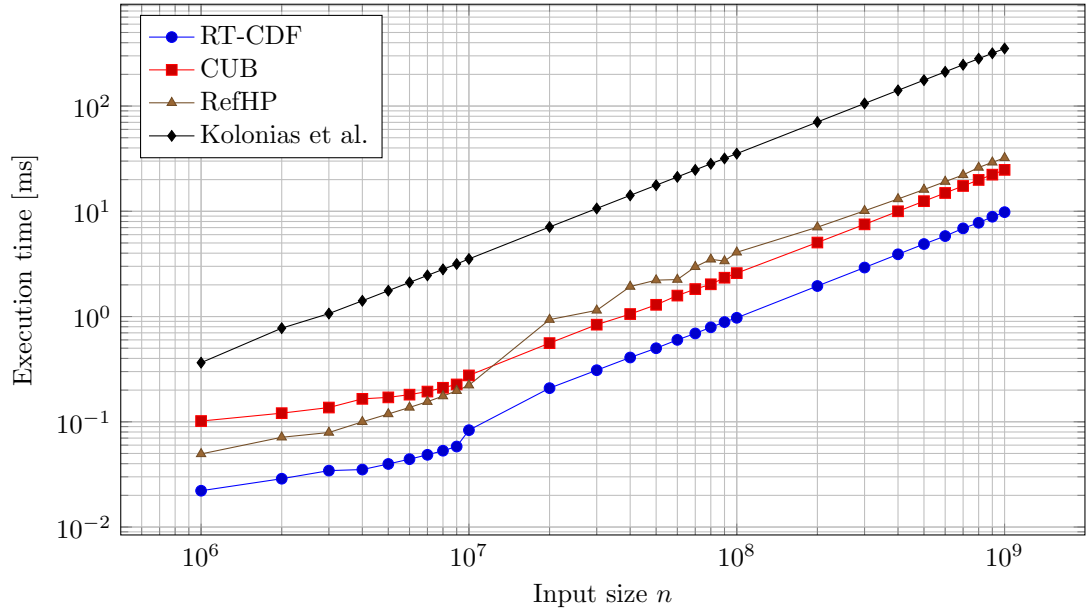

%% file: figures/fig_runtime_r1024_all_equal.tex
\begin{figure}[H]
\centering
\begin{tikzpicture}
\begin{axis}[
width=0.90\linewidth,height=0.54\linewidth,
xmode=log,ymode=log,log basis x=10,log basis y=10,
xlabel={Input size $n$},ylabel={Execution time [ms]},
grid=both,
unbounded coords=jump,
legend style={at={(0.02,0.98)},anchor=north west,fill=white,draw=black},
legend cell align={left}
]
\addplot+[mark=*] table[x=n,y=rtcdf,col sep=comma] {plotdata/runtime_R1024_all_equal_rtx4090.dat};
\addlegendentry{RT-CDF}
\addplot+[mark=square*] table[x=n,y=cub,col sep=comma] {plotdata/runtime_R1024_all_equal_rtx4090.dat};
\addlegendentry{CUB}
\addplot+[mark=triangle*] table[x=n,y=refhp,col sep=comma] {plotdata/runtime_R1024_all_equal_rtx4090.dat};
\addlegendentry{RefHP}
\addplot+[mark=diamond*] table[x=n,y=kolonias,col sep=comma] {plotdata/runtime_R1024_all_equal_rtx4090.dat};
\addlegendentry{Kolonias et al.}
\end{axis}
\end{tikzpicture}
\caption{Execution time for all-equal inputs ($R=2^{10}$)}
\label{fig:runtime-r1024-all-equal}
\end{figure}
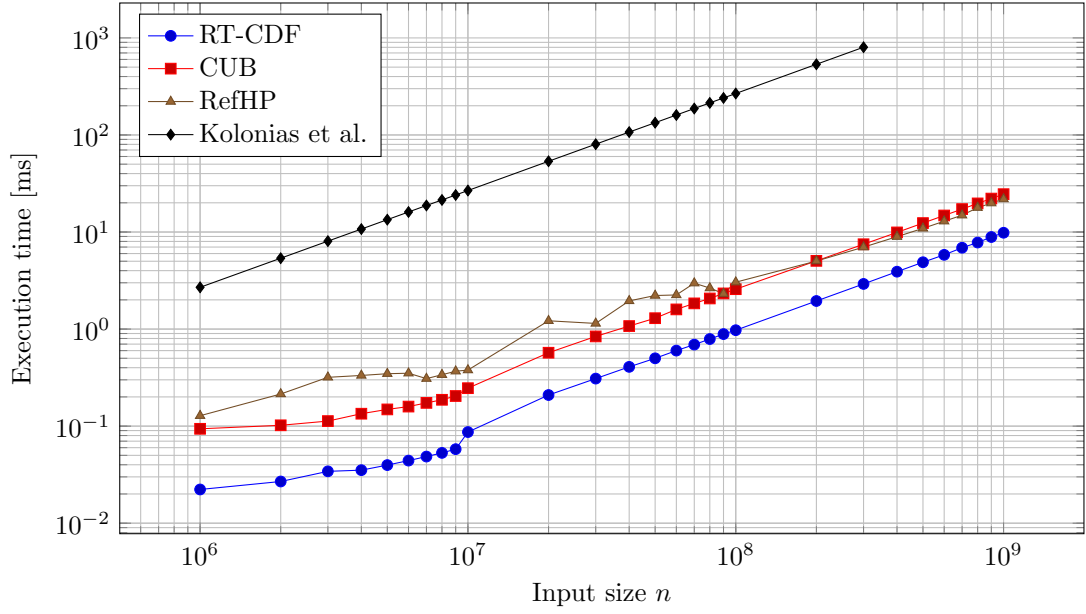

%% file: figures/fig_runtime_r131072_uniform.tex
\begin{figure}[H]
\centering
\begin{tikzpicture}
\begin{axis}[
width=0.90\linewidth,height=0.54\linewidth,
xmode=log,ymode=log,log basis x=10,log basis y=10,
xlabel={Input size $n$},ylabel={Execution time [ms]},
grid=both,
unbounded coords=jump,
legend style={at={(0.02,0.98)},anchor=north west,fill=white,draw=black},
legend cell align={left}
]
\addplot+[mark=*] table[x=n,y=rtcdf,col sep=comma] {plotdata/runtime_R131072_uniform_rtx4090.dat};
\addlegendentry{RT-CDF}
\addplot+[mark=square*] table[x=n,y=cub,col sep=comma] {plotdata/runtime_R131072_uniform_rtx4090.dat};
\addlegendentry{CUB}
\addplot+[mark=triangle*] table[x=n,y=refhp,col sep=comma] {plotdata/runtime_R131072_uniform_rtx4090.dat};
\addlegendentry{RefHP}
\addplot+[mark=diamond*] table[x=n,y=kolonias,col sep=comma] {plotdata/runtime_R131072_uniform_rtx4090.dat};
\addlegendentry{Kolonias et al.}
\end{axis}
\end{tikzpicture}
\caption{Execution time for uniformly distributed inputs ($R=2^{17}$)}
\label{fig:runtime-r131072-uniform}
\end{figure}
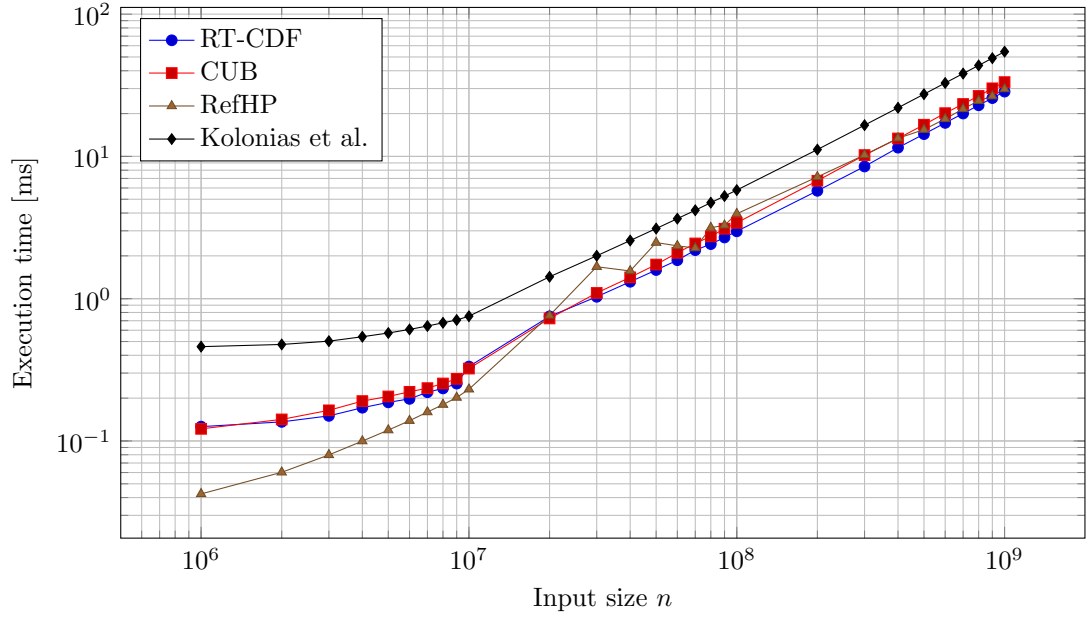

%% file: figures/fig_runtime_r131072_gaussian.tex
\begin{figure}[H]
\centering
\begin{tikzpicture}
\begin{axis}[
width=0.90\linewidth,height=0.54\linewidth,
xmode=log,ymode=log,log basis x=10,log basis y=10,
xlabel={Input size $n$},ylabel={Execution time [ms]},
grid=both,
unbounded coords=jump,
legend style={at={(0.02,0.98)},anchor=north west,fill=white,draw=black},
legend cell align={left}
]
\addplot+[mark=*] table[x=n,y=rtcdf,col sep=comma] {plotdata/runtime_R131072_gaussian_rtx4090.dat};
\addlegendentry{RT-CDF}
\addplot+[mark=square*] table[x=n,y=cub,col sep=comma] {plotdata/runtime_R131072_gaussian_rtx4090.dat};
\addlegendentry{CUB}
\addplot+[mark=triangle*] table[x=n,y=refhp,col sep=comma] {plotdata/runtime_R131072_gaussian_rtx4090.dat};
\addlegendentry{RefHP}
\addplot+[mark=diamond*] table[x=n,y=kolonias,col sep=comma] {plotdata/runtime_R131072_gaussian_rtx4090.dat};
\addlegendentry{Kolonias et al.}
\end{axis}
\end{tikzpicture}
\caption{Execution time for normally distributed inputs ($R=2^{17}$)}
\label{fig:runtime-r131072-gaussian}
\end{figure}
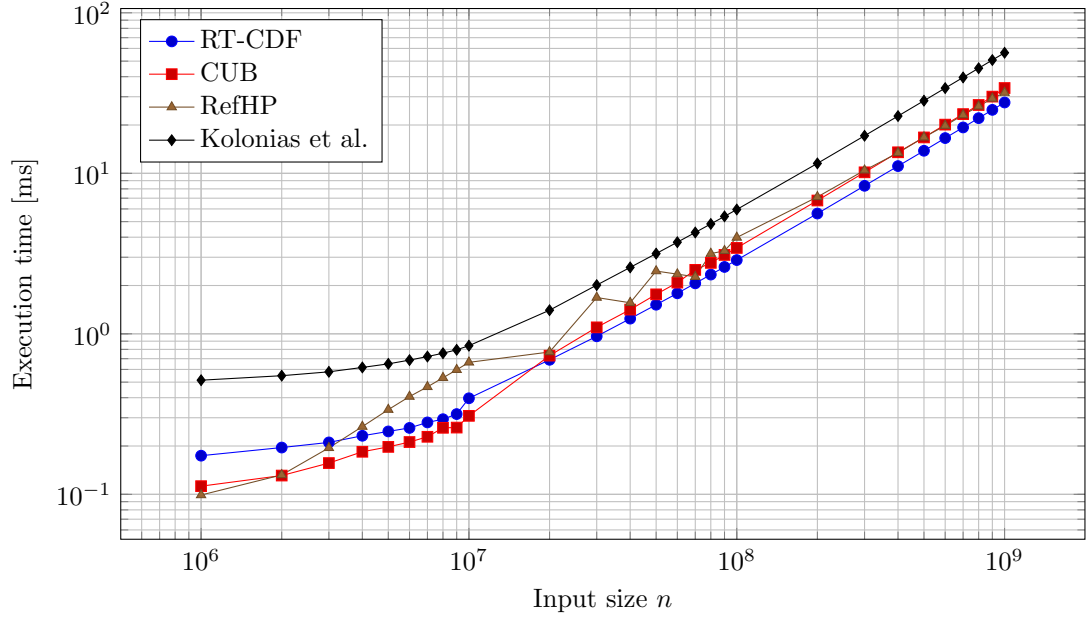

%% file: figures/fig_runtime_r131072_all_equal.tex
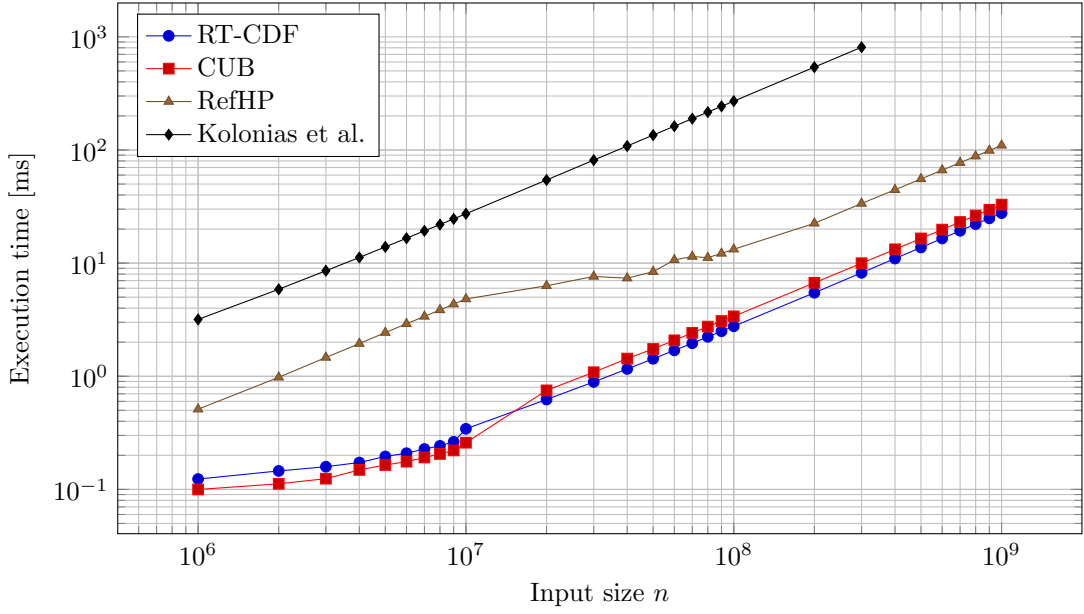
\begin{figure}[H]
\centering
\begin{tikzpicture}
\begin{axis}[
width=0.90\linewidth,height=0.54\linewidth,
xmode=log,ymode=log,log basis x=10,log basis y=10,
xlabel={Input size $n$},ylabel={Execution time [ms]},
grid=both,
unbounded coords=jump,
legend style={at={(0.02,0.98)},anchor=north west,fill=white,draw=black},
legend cell align={left}
]
\addplot+[mark=*] table[x=n,y=rtcdf,col sep=comma] {plotdata/runtime_R131072_all_equal_rtx4090.dat};
\addlegendentry{RT-CDF}
\addplot+[mark=square*] table[x=n,y=cub,col sep=comma] {plotdata/runtime_R131072_all_equal_rtx4090.dat};
\addlegendentry{CUB}
\addplot+[mark=triangle*] table[x=n,y=refhp,col sep=comma] {plotdata/runtime_R131072_all_equal_rtx4090.dat};
\addlegendentry{RefHP}
\addplot+[mark=diamond*] table[x=n,y=kolonias,col sep=comma] {plotdata/runtime_R131072_all_equal_rtx4090.dat};
\addlegendentry{Kolonias et al.}
\end{axis}
\end{tikzpicture}
\caption{Execution time for all-equal inputs ($R=2^{17}$)}
\label{fig:runtime-r131072-all-equal}
\end{figure}

%% file: figures/fig_runtime_r262144_uniform.tex
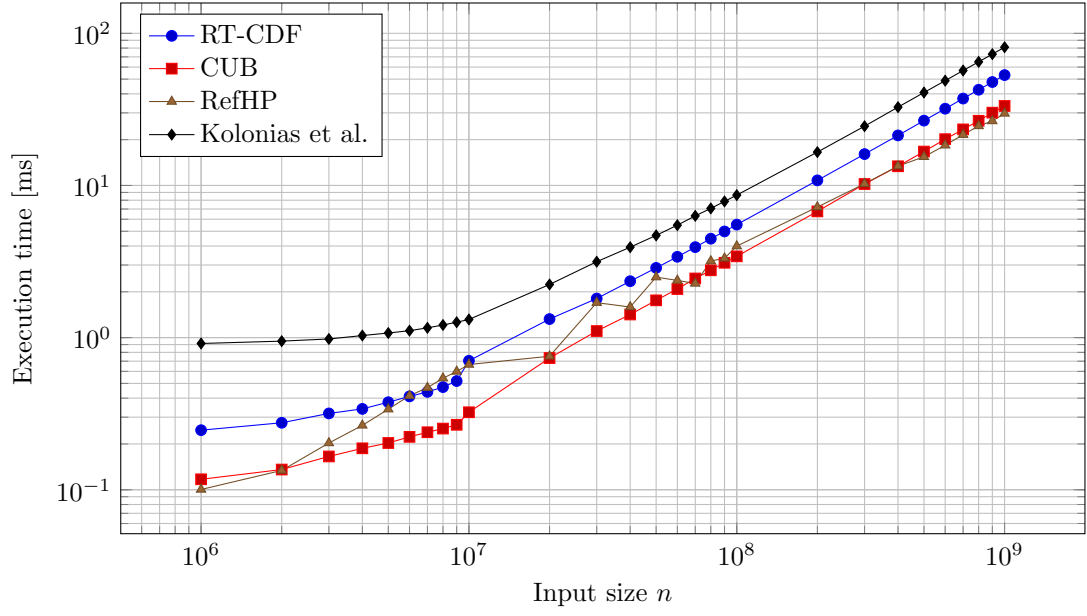
\begin{figure}[H]
\centering
\begin{tikzpicture}
\begin{axis}[
width=0.90\linewidth,height=0.54\linewidth,
xmode=log,ymode=log,log basis x=10,log basis y=10,
xlabel={Input size $n$},ylabel={Execution time [ms]},
grid=both,
unbounded coords=jump,
legend style={at={(0.02,0.98)},anchor=north west,fill=white,draw=black},
legend cell align={left}
]
\addplot+[mark=*] table[x=n,y=rtcdf,col sep=comma] {plotdata/runtime_R262144_uniform_rtx4090.dat};
\addlegendentry{RT-CDF}
\addplot+[mark=square*] table[x=n,y=cub,col sep=comma] {plotdata/runtime_R262144_uniform_rtx4090.dat};
\addlegendentry{CUB}
\addplot+[mark=triangle*] table[x=n,y=refhp,col sep=comma] {plotdata/runtime_R262144_uniform_rtx4090.dat};
\addlegendentry{RefHP}
\addplot+[mark=diamond*] table[x=n,y=kolonias,col sep=comma] {plotdata/runtime_R262144_uniform_rtx4090.dat};
\addlegendentry{Kolonias et al.}
\end{axis}
\end{tikzpicture}
\caption{Execution time for uniformly distributed inputs ($R=2^{18}$)}
\label{fig:runtime-r262144-uniform}
\end{figure}

%% file: figures/fig_runtime_r262144_gaussian.tex
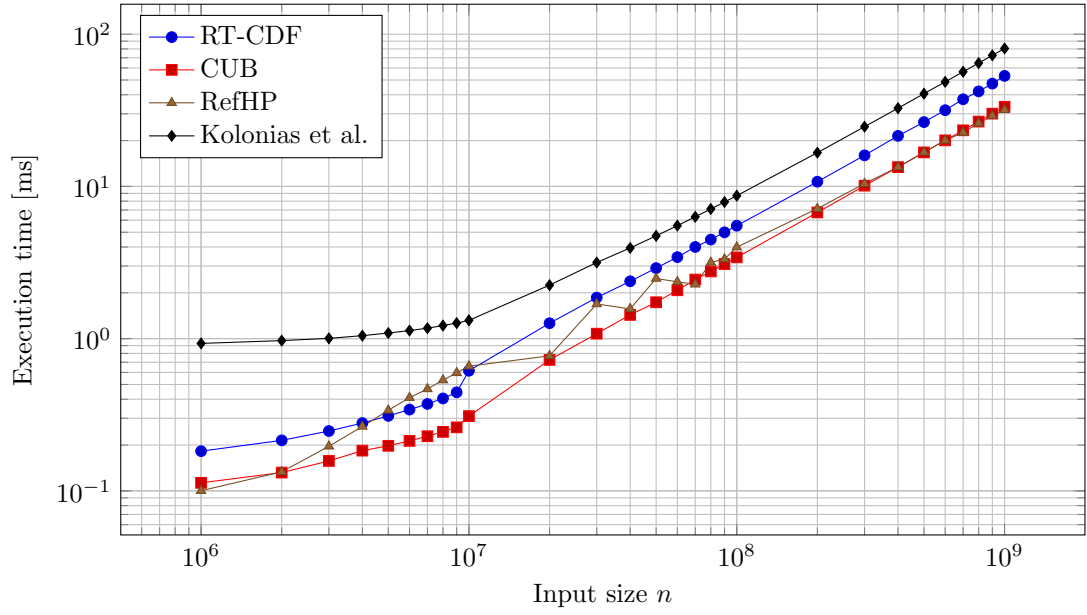
\begin{figure}[H]
\centering
\begin{tikzpicture}
\begin{axis}[
width=0.90\linewidth,height=0.54\linewidth,
xmode=log,ymode=log,log basis x=10,log basis y=10,
xlabel={Input size $n$},ylabel={Execution time [ms]},
grid=both,
unbounded coords=jump,
legend style={at={(0.02,0.98)},anchor=north west,fill=white,draw=black},
legend cell align={left}
]
\addplot+[mark=*] table[x=n,y=rtcdf,col sep=comma] {plotdata/runtime_R262144_gaussian_rtx4090.dat};
\addlegendentry{RT-CDF}
\addplot+[mark=square*] table[x=n,y=cub,col sep=comma] {plotdata/runtime_R262144_gaussian_rtx4090.dat};
\addlegendentry{CUB}
\addplot+[mark=triangle*] table[x=n,y=refhp,col sep=comma] {plotdata/runtime_R262144_gaussian_rtx4090.dat};
\addlegendentry{RefHP}
\addplot+[mark=diamond*] table[x=n,y=kolonias,col sep=comma] {plotdata/runtime_R262144_gaussian_rtx4090.dat};
\addlegendentry{Kolonias et al.}
\end{axis}
\end{tikzpicture}
\caption{Execution time for normally distributed inputs ($R=2^{18}$)}
\label{fig:runtime-r262144-gaussian}
\end{figure}

%% file: figures/fig_runtime_r262144_all_equal.tex
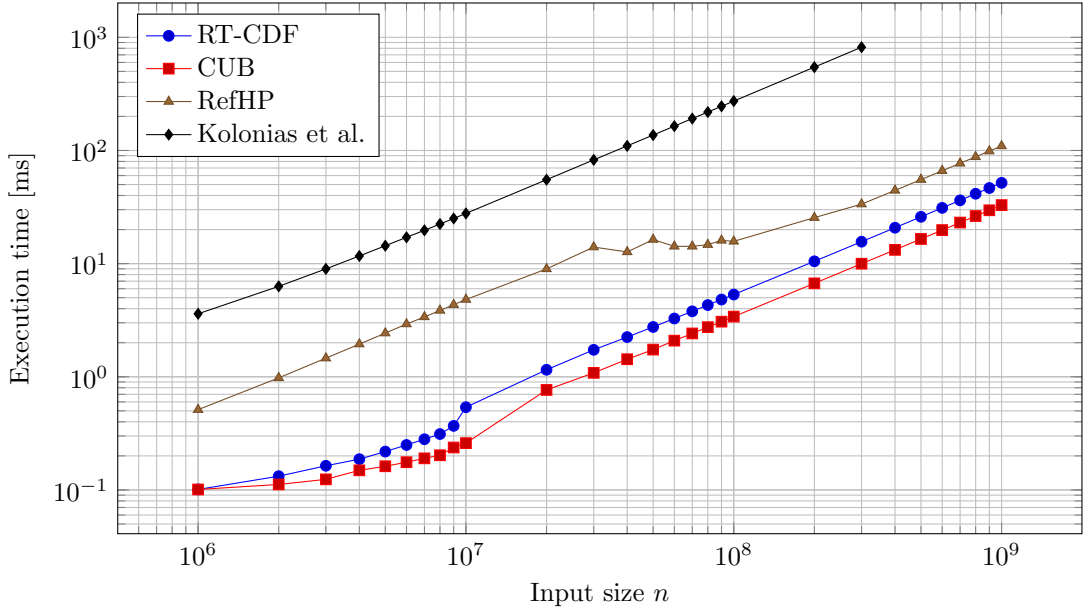
\begin{figure}[H]
\centering
\begin{tikzpicture}
\begin{axis}[
width=0.90\linewidth,height=0.54\linewidth,
xmode=log,ymode=log,log basis x=10,log basis y=10,
xlabel={Input size $n$},ylabel={Execution time [ms]},
grid=both,
unbounded coords=jump,
legend style={at={(0.02,0.98)},anchor=north west,fill=white,draw=black},
legend cell align={left}
]
\addplot+[mark=*] table[x=n,y=rtcdf,col sep=comma] {plotdata/runtime_R262144_all_equal_rtx4090.dat};
\addlegendentry{RT-CDF}
\addplot+[mark=square*] table[x=n,y=cub,col sep=comma] {plotdata/runtime_R262144_all_equal_rtx4090.dat};
\addlegendentry{CUB}
\addplot+[mark=triangle*] table[x=n,y=refhp,col sep=comma] {plotdata/runtime_R262144_all_equal_rtx4090.dat};
\addlegendentry{RefHP}
\addplot+[mark=diamond*] table[x=n,y=kolonias,col sep=comma] {plotdata/runtime_R262144_all_equal_rtx4090.dat};
\addlegendentry{Kolonias et al.}
\end{axis}
\end{tikzpicture}
\caption{Execution time for all-equal inputs ($R=2^{18}$)}
\label{fig:runtime-r262144-all-equal}
\end{figure}

%% file: figures/fig_speedup_uniform.tex
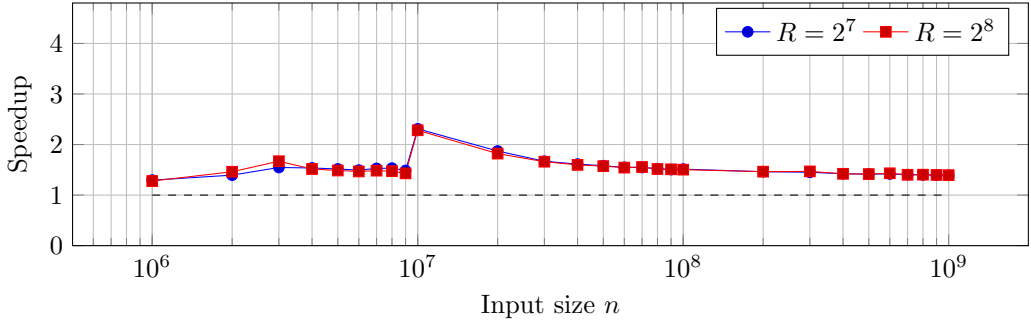
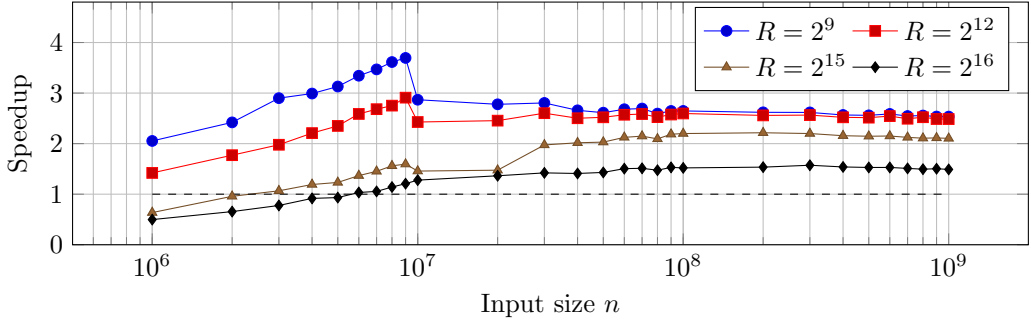
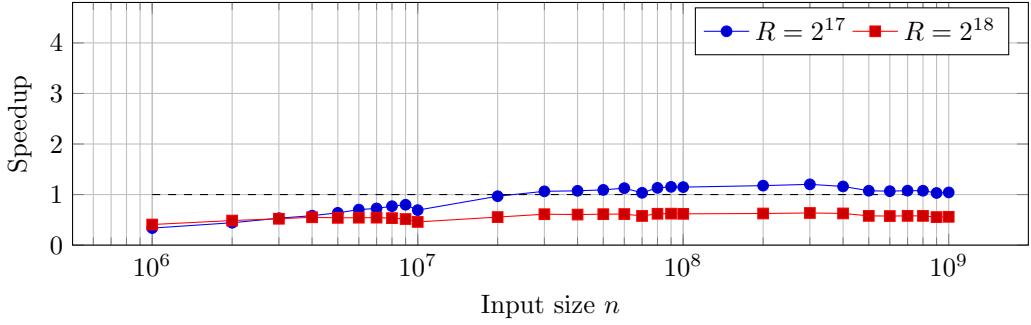
\begin{figure}[H]
\centering
\begin{subfigure}[t]{0.94\linewidth}
\centering
\begin{tikzpicture}
\begin{axis}[
width=0.95\linewidth,height=0.32\linewidth,
xmode=log,log basis x=10,
xlabel={Input size $n$},ylabel={Speedup},
ymin=0,ymax=4.8,
grid=both,
legend style={at={(0.98,0.98)},anchor=north east,legend columns=2,fill=white,draw=black},
legend cell align={left},
ytick={0,1,2,3,4}
]
\addplot+[mark=*] table[x=n,y=r128,col sep=comma] {plotdata/speedup_uniform_rtx4090.dat};
\addlegendentry{$R=2^7$}
\addplot+[mark=square*] table[x=n,y=r256,col sep=comma] {plotdata/speedup_uniform_rtx4090.dat};
\addlegendentry{$R=2^8$}
\addplot[dashed,forget plot] coordinates {(1000000,1) (1000000000,1)};
\end{axis}
\end{tikzpicture}
\caption{Ranges requiring one CUB pass ($R\le 2^8$)}
\end{subfigure}
\vspace{0.5em}
\begin{subfigure}[t]{0.94\linewidth}
\centering
\begin{tikzpicture}
\begin{axis}[
width=0.95\linewidth,height=0.32\linewidth,
xmode=log,log basis x=10,
xlabel={Input size $n$},ylabel={Speedup},
ymin=0,ymax=4.8,
grid=both,
legend style={at={(0.98,0.98)},anchor=north east,legend columns=2,fill=white,draw=black},
legend cell align={left},
ytick={0,1,2,3,4}
]
\addplot+[mark=*] table[x=n,y=r512,col sep=comma] {plotdata/speedup_uniform_rtx4090.dat};
\addlegendentry{$R=2^9$}
\addplot+[mark=square*] table[x=n,y=r4096,col sep=comma] {plotdata/speedup_uniform_rtx4090.dat};
\addlegendentry{$R=2^{12}$}
\addplot+[mark=triangle*] table[x=n,y=r32768,col sep=comma] {plotdata/speedup_uniform_rtx4090.dat};
\addlegendentry{$R=2^{15}$}
\addplot+[mark=diamond*] table[x=n,y=r65536,col sep=comma] {plotdata/speedup_uniform_rtx4090.dat};
\addlegendentry{$R=2^{16}$}
\addplot[dashed,forget plot] coordinates {(1000000,1) (1000000000,1)};
\end{axis}
\end{tikzpicture}
\caption{Representative ranges requiring two CUB passes ($2^9\le R\le 2^{16}$)}
\end{subfigure}
\vspace{0.5em}
\begin{subfigure}[t]{0.94\linewidth}
\centering
\begin{tikzpicture}
\begin{axis}[
width=0.95\linewidth,height=0.32\linewidth,
xmode=log,log basis x=10,
xlabel={Input size $n$},ylabel={Speedup},
ymin=0,ymax=4.8,
grid=both,
legend style={at={(0.98,0.98)},anchor=north east,legend columns=2,fill=white,draw=black},
legend cell align={left},
ytick={0,1,2,3,4}
]
\addplot+[mark=*] table[x=n,y=r131072,col sep=comma] {plotdata/speedup_uniform_rtx4090.dat};
\addlegendentry{$R=2^{17}$}
\addplot+[mark=square*] table[x=n,y=r262144,col sep=comma] {plotdata/speedup_uniform_rtx4090.dat};
\addlegendentry{$R=2^{18}$}
\addplot[dashed,forget plot] coordinates {(1000000,1) (1000000000,1)};
\end{axis}
\end{tikzpicture}
\caption{Ranges requiring three CUB passes ($2^{17}\le R\le 2^{18}$)}
\end{subfigure}
\caption{RT-CDF speedup over the fastest baseline for uniformly distributed inputs}
\label{fig:speedup-uniform}
\end{figure}

%% file: figures/fig_speedup_gaussian.tex
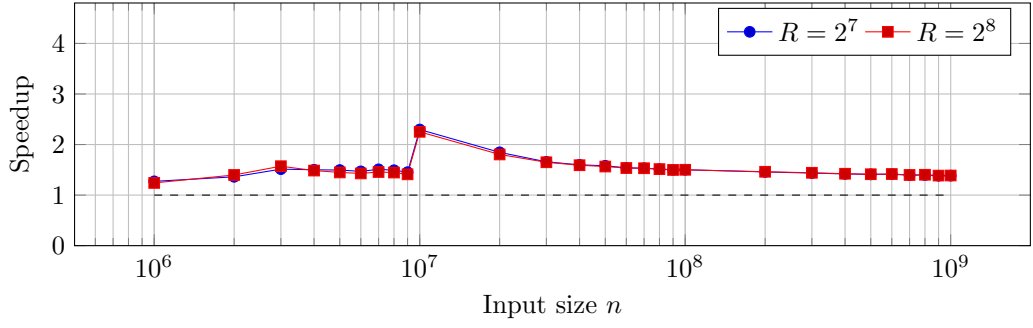
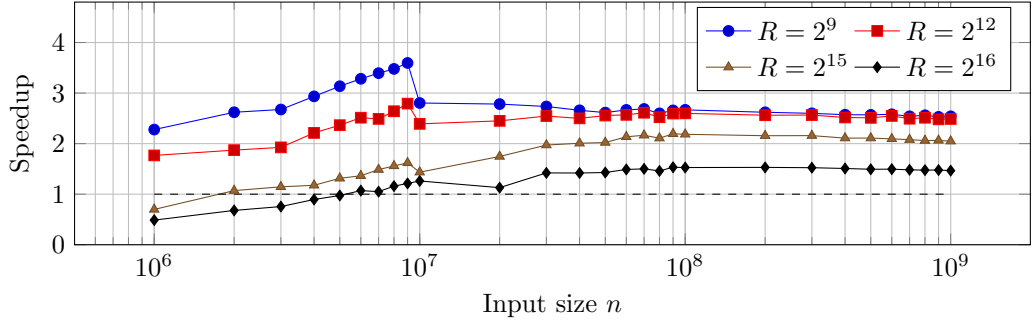
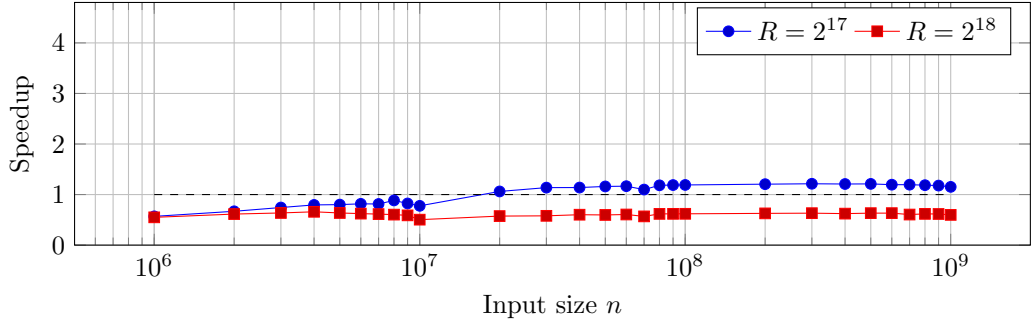
\begin{figure}[H]
\centering
\begin{subfigure}[t]{0.94\linewidth}
\centering
\begin{tikzpicture}
\begin{axis}[
width=0.95\linewidth,height=0.32\linewidth,
xmode=log,log basis x=10,
xlabel={Input size $n$},ylabel={Speedup},
ymin=0,ymax=4.8,
grid=both,
legend style={at={(0.98,0.98)},anchor=north east,legend columns=2,fill=white,draw=black},
legend cell align={left},
ytick={0,1,2,3,4}
]
\addplot+[mark=*] table[x=n,y=r128,col sep=comma] {plotdata/speedup_gaussian_rtx4090.dat};
\addlegendentry{$R=2^7$}
\addplot+[mark=square*] table[x=n,y=r256,col sep=comma] {plotdata/speedup_gaussian_rtx4090.dat};
\addlegendentry{$R=2^8$}
\addplot[dashed,forget plot] coordinates {(1000000,1) (1000000000,1)};
\end{axis}
\end{tikzpicture}
\caption{Ranges requiring one CUB pass ($R\le 2^8$)}
\end{subfigure}
\vspace{0.5em}
\begin{subfigure}[t]{0.94\linewidth}
\centering
\begin{tikzpicture}
\begin{axis}[
width=0.95\linewidth,height=0.32\linewidth,
xmode=log,log basis x=10,
xlabel={Input size $n$},ylabel={Speedup},
ymin=0,ymax=4.8,
grid=both,
legend style={at={(0.98,0.98)},anchor=north east,legend columns=2,fill=white,draw=black},
legend cell align={left},
ytick={0,1,2,3,4}
]
\addplot+[mark=*] table[x=n,y=r512,col sep=comma] {plotdata/speedup_gaussian_rtx4090.dat};
\addlegendentry{$R=2^9$}
\addplot+[mark=square*] table[x=n,y=r4096,col sep=comma] {plotdata/speedup_gaussian_rtx4090.dat};
\addlegendentry{$R=2^{12}$}
\addplot+[mark=triangle*] table[x=n,y=r32768,col sep=comma] {plotdata/speedup_gaussian_rtx4090.dat};
\addlegendentry{$R=2^{15}$}
\addplot+[mark=diamond*] table[x=n,y=r65536,col sep=comma] {plotdata/speedup_gaussian_rtx4090.dat};
\addlegendentry{$R=2^{16}$}
\addplot[dashed,forget plot] coordinates {(1000000,1) (1000000000,1)};
\end{axis}
\end{tikzpicture}
\caption{Representative ranges requiring two CUB passes ($2^9\le R\le 2^{16}$)}
\end{subfigure}
\vspace{0.5em}
\begin{subfigure}[t]{0.94\linewidth}
\centering
\begin{tikzpicture}
\begin{axis}[
width=0.95\linewidth,height=0.32\linewidth,
xmode=log,log basis x=10,
xlabel={Input size $n$},ylabel={Speedup},
ymin=0,ymax=4.8,
grid=both,
legend style={at={(0.98,0.98)},anchor=north east,legend columns=2,fill=white,draw=black},
legend cell align={left},
ytick={0,1,2,3,4}
]
\addplot+[mark=*] table[x=n,y=r131072,col sep=comma] {plotdata/speedup_gaussian_rtx4090.dat};
\addlegendentry{$R=2^{17}$}
\addplot+[mark=square*] table[x=n,y=r262144,col sep=comma] {plotdata/speedup_gaussian_rtx4090.dat};
\addlegendentry{$R=2^{18}$}
\addplot[dashed,forget plot] coordinates {(1000000,1) (1000000000,1)};
\end{axis}
\end{tikzpicture}
\caption{Ranges requiring three CUB passes ($2^{17}\le R\le 2^{18}$)}
\end{subfigure}
\caption{RT-CDF speedup over the fastest baseline for normally distributed inputs}
\label{fig:speedup-gaussian}
\end{figure}

%% file: figures/fig_speedup_all_equal.tex
\begin{figure}[H]
\centering
\begin{subfigure}[t]{0.94\linewidth}
\centering
\begin{tikzpicture}
\begin{axis}[
width=0.95\linewidth,height=0.32\linewidth,
xmode=log,log basis x=10,
xlabel={Input size $n$},ylabel={Speedup},
ymin=0,ymax=4.8,
grid=both,
legend style={at={(0.98,0.98)},anchor=north east,legend columns=2,fill=white,draw=black},
legend cell align={left},
ytick={0,1,2,3,4}
]
\addplot+[mark=*] table[x=n,y=r128,col sep=comma] {plotdata/speedup_all_equal_rtx4090.dat};
\addlegendentry{$R=2^7$}
\addplot+[mark=square*] table[x=n,y=r256,col sep=comma] {plotdata/speedup_all_equal_rtx4090.dat};
\addlegendentry{$R=2^8$}
\addplot[dashed,forget plot] coordinates {(1000000,1) (1000000000,1)};
\end{axis}
\end{tikzpicture}
\caption{Ranges requiring one CUB pass ($R\le 2^8$)}
\end{subfigure}
\vspace{0.5em}
\begin{subfigure}[t]{0.94\linewidth}
\centering
\begin{tikzpicture}
\begin{axis}[
width=0.95\linewidth,height=0.32\linewidth,
xmode=log,log basis x=10,
xlabel={Input size $n$},ylabel={Speedup},
ymin=0,ymax=4.8,
grid=both,
legend style={at={(0.98,0.98)},anchor=north east,legend columns=2,fill=white,draw=black},
legend cell align={left},
ytick={0,1,2,3,4}
]
\addplot+[mark=*] table[x=n,y=r512,col sep=comma] {plotdata/speedup_all_equal_rtx4090.dat};
\addlegendentry{$R=2^9$}
\addplot+[mark=square*] table[x=n,y=r4096,col sep=comma] {plotdata/speedup_all_equal_rtx4090.dat};
\addlegendentry{$R=2^{12}$}
\addplot+[mark=triangle*] table[x=n,y=r32768,col sep=comma] {plotdata/speedup_all_equal_rtx4090.dat};
\addlegendentry{$R=2^{15}$}
\addplot+[mark=diamond*] table[x=n,y=r65536,col sep=comma] {plotdata/speedup_all_equal_rtx4090.dat};
\addlegendentry{$R=2^{16}$}
\addplot[dashed,forget plot] coordinates {(1000000,1) (1000000000,1)};
\end{axis}
\end{tikzpicture}
\caption{Representative ranges requiring two CUB passes ($2^9\le R\le 2^{16}$)}
\end{subfigure}
\vspace{0.5em}
\begin{subfigure}[t]{0.94\linewidth}
\centering
\begin{tikzpicture}
\begin{axis}[
width=0.95\linewidth,height=0.32\linewidth,
xmode=log,log basis x=10,
xlabel={Input size $n$},ylabel={Speedup},
ymin=0,ymax=4.8,
grid=both,
legend style={at={(0.98,0.98)},anchor=north east,legend columns=2,fill=white,draw=black},
legend cell align={left},
ytick={0,1,2,3,4}
]
\addplot+[mark=*] table[x=n,y=r131072,col sep=comma] {plotdata/speedup_all_equal_rtx4090.dat};
\addlegendentry{$R=2^{17}$}
\addplot+[mark=square*] table[x=n,y=r262144,col sep=comma] {plotdata/speedup_all_equal_rtx4090.dat};
\addlegendentry{$R=2^{18}$}
\addplot[dashed,forget plot] coordinates {(1000000,1) (1000000000,1)};
\end{axis}
\end{tikzpicture}
\caption{Ranges requiring three CUB passes ($2^{17}\le R\le 2^{18}$)}
\end{subfigure}
\caption{RT-CDF speedup over the fastest baseline for all-equal inputs}
\label{fig:speedup-all-equal}
\end{figure}
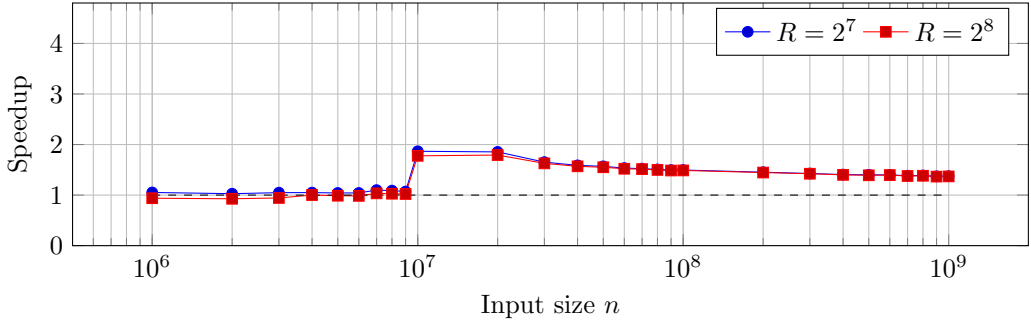
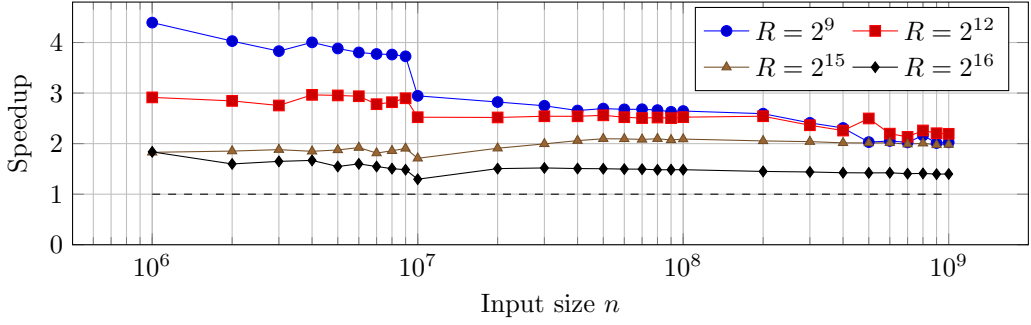
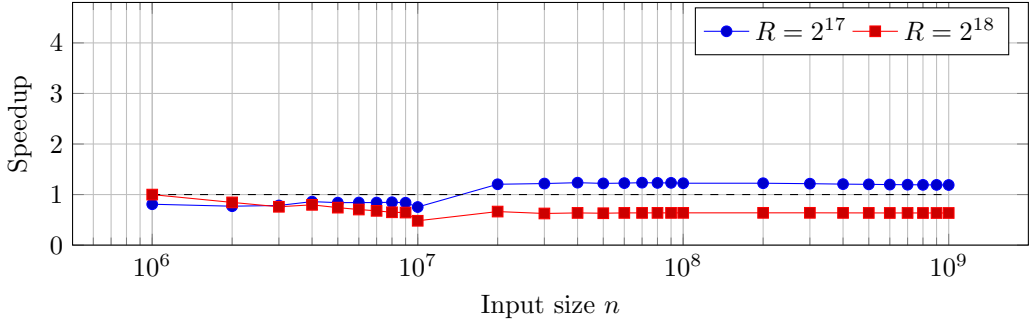

%% file: tables/table_winner_map_uniform.tex
\begin{table}[H]
  \centering
  \caption{Fastest methods at representative input sizes (uniform distribution)}
  \label{tab:winner-map-uniform}
  \begin{tabular}{ccccc}
    \toprule
    & \multicolumn{4}{c}{Input size $n$} \\
    \cmidrule(lr){2-5}
    Range size $R$
      & $10^6$
      & $10^7$
      & $10^8$
      & $10^9$ \\
    \midrule
    $2^7$    & RT-CDF & RT-CDF & RT-CDF & RT-CDF \\
    $2^8$    & RT-CDF & RT-CDF & RT-CDF & RT-CDF \\
    $2^9$    & RT-CDF & RT-CDF & RT-CDF & RT-CDF \\
    $2^{10}$ & RT-CDF & RT-CDF & RT-CDF & RT-CDF \\
    $2^{11}$ & RT-CDF & RT-CDF & RT-CDF & RT-CDF \\
    $2^{12}$ & RT-CDF & RT-CDF & RT-CDF & RT-CDF \\
    $2^{13}$ & RefHP  & RT-CDF & RT-CDF & RT-CDF \\
    $2^{14}$ & RefHP  & RT-CDF & RT-CDF & RT-CDF \\
    $2^{15}$ & RefHP  & RT-CDF & RT-CDF & RT-CDF \\
    $2^{16}$ & RefHP  & RT-CDF & RT-CDF & RT-CDF \\
    $2^{17}$ & RefHP  & RefHP  & RT-CDF & RT-CDF \\
    $2^{18}$ & RefHP  & CUB    & CUB    & RefHP  \\
    \bottomrule
  \end{tabular}
\end{table}

%% file: tables/table_winner_map_gaussian.tex
\begin{table}[H]
  \centering
  \caption{Fastest methods at representative input sizes (normal distribution)}
  \label{tab:winner-map-gaussian}
  \begin{tabular}{ccccc}
    \toprule
    & \multicolumn{4}{c}{Input size $n$} \\
    \cmidrule(lr){2-5}
    Range size $R$
      & $10^6$
      & $10^7$
      & $10^8$
      & $10^9$ \\
    \midrule
    $2^7$    & RT-CDF & RT-CDF & RT-CDF & RT-CDF \\
    $2^8$    & RT-CDF & RT-CDF & RT-CDF & RT-CDF \\
    $2^9$    & RT-CDF & RT-CDF & RT-CDF & RT-CDF \\
    $2^{10}$ & RT-CDF & RT-CDF & RT-CDF & RT-CDF \\
    $2^{11}$ & RT-CDF & RT-CDF & RT-CDF & RT-CDF \\
    $2^{12}$ & RT-CDF & RT-CDF & RT-CDF & RT-CDF \\
    $2^{13}$ & RT-CDF & RT-CDF & RT-CDF & RT-CDF \\
    $2^{14}$ & RT-CDF & RT-CDF & RT-CDF & RT-CDF \\
    $2^{15}$ & RefHP  & RT-CDF & RT-CDF & RT-CDF \\
    $2^{16}$ & RefHP  & RT-CDF & RT-CDF & RT-CDF \\
    $2^{17}$ & RefHP  & CUB    & RT-CDF & RT-CDF \\
    $2^{18}$ & RefHP  & CUB    & CUB    & RefHP  \\
    \bottomrule
  \end{tabular}
\end{table}

%% file: tables/table_winner_map_all_equal.tex
\begin{table}[H]
  \centering
  \caption{Fastest methods at representative input sizes (all-equal input)}
  \label{tab:winner-map-all-equal}
  \begin{tabular}{ccccc}
    \toprule
    & \multicolumn{4}{c}{Input size $n$} \\
    \cmidrule(lr){2-5}
    Range size $R$
      & $10^6$
      & $10^7$
      & $10^8$
      & $10^9$ \\
    \midrule
    $2^7$    & RT-CDF & RT-CDF & RT-CDF & RT-CDF \\
    $2^8$    & CUB    & RT-CDF & RT-CDF & RT-CDF \\
    $2^9$    & RT-CDF & RT-CDF & RT-CDF & RT-CDF \\
    $2^{10}$ & RT-CDF & RT-CDF & RT-CDF & RT-CDF \\
    $2^{11}$ & RT-CDF & RT-CDF & RT-CDF & RT-CDF \\
    $2^{12}$ & RT-CDF & RT-CDF & RT-CDF & RT-CDF \\
    $2^{13}$ & RT-CDF & RT-CDF & RT-CDF & RT-CDF \\
    $2^{14}$ & RT-CDF & RT-CDF & RT-CDF & RT-CDF \\
    $2^{15}$ & RT-CDF & RT-CDF & RT-CDF & RT-CDF \\
    $2^{16}$ & RT-CDF & RT-CDF & RT-CDF & RT-CDF \\
    $2^{17}$ & CUB    & CUB    & RT-CDF & RT-CDF \\
    $2^{18}$ & CUB    & CUB    & CUB    & CUB    \\
    \bottomrule
  \end{tabular}
\end{table}

%% file: tables/table_runtime_four_n_uniform.tex
\begingroup
\small
\setlength{\LTleft}{\fill}
\setlength{\LTright}{\fill}
\setlength{\tabcolsep}{7pt}
\renewcommand{\arraystretch}{1.05}
\begin{longtable}{lrrrr}
\caption{Execution times of the methods at representative input sizes (uniform distribution)}
\label{tab:runtime-four-n-uniform}\\
\toprule
$R$ & RT-CDF [ms] & CUB [ms] & RefHP [ms] & Kolonias et al. [ms] \\
\midrule
\endfirsthead

\multicolumn{5}{c}{Table \thetable\ (continued)}\\
\toprule
$R$ & RT-CDF [ms] & CUB [ms] & RefHP [ms] & Kolonias et al. [ms] \\
\midrule
\endhead

\midrule
\multicolumn{5}{r}{Continued on the next page}\\
\endfoot

\bottomrule
\endlastfoot

\multicolumn{5}{l}{\textbf{$n=10^{6}$}} \\
\addlinespace[1pt]
$2^{7}$ & \textbf{0.01772} & 0.02296 & 0.04203 & 0.7402 \\
$2^{8}$ & \textbf{0.01998} & 0.02553 & 0.04174 & 0.7446 \\
$2^{9}$ & \textbf{0.02095} & 0.1034 & 0.04301 & 0.3795 \\
$2^{10}$ & \textbf{0.02190} & 0.1028 & 0.04327 & 0.1989 \\
$2^{11}$ & \textbf{0.02700} & 0.1027 & 0.04134 & 0.1128 \\
$2^{12}$ & \textbf{0.03053} & 0.1025 & 0.04337 & 0.07029 \\
$2^{13}$ & 0.04580 & 0.1042 & \textbf{0.04311} & 0.06548 \\
$2^{14}$ & 0.05649 & 0.1053 & \textbf{0.04325} & 0.09759 \\
$2^{15}$ & 0.06625 & 0.1056 & \textbf{0.04215} & 0.1665 \\
$2^{16}$ & 0.08320 & 0.1067 & \textbf{0.04139} & 0.2489 \\
$2^{17}$ & 0.1261 & 0.1217 & \textbf{0.04245} & 0.4601 \\
$2^{18}$ & 0.2464 & 0.1172 & \textbf{0.1004} & 0.9149 \\
\addlinespace[2pt]
\midrule
\multicolumn{5}{l}{\textbf{$n=10^{7}$}} \\
\addlinespace[1pt]
$2^{7}$ & \textbf{0.07101} & 0.1641 & 0.2198 & 7.327 \\
$2^{8}$ & \textbf{0.07424} & 0.1694 & 0.2218 & 7.345 \\
$2^{9}$ & \textbf{0.07891} & 0.2742 & 0.2263 & 3.694 \\
$2^{10}$ & \textbf{0.08275} & 0.2737 & 0.2254 & 1.871 \\
$2^{11}$ & \textbf{0.08786} & 0.2738 & 0.2278 & 0.9625 \\
$2^{12}$ & \textbf{0.09407} & 0.2735 & 0.2284 & 0.5096 \\
$2^{13}$ & \textbf{0.1105} & 0.2766 & 0.2278 & 0.3295 \\
$2^{14}$ & \textbf{0.1245} & 0.2784 & 0.2274 & 0.3526 \\
$2^{15}$ & \textbf{0.1571} & 0.2822 & 0.2289 & 0.4213 \\
$2^{16}$ & \textbf{0.2172} & 0.2776 & 0.7029 & 0.5276 \\
$2^{17}$ & 0.3335 & 0.3227 & \textbf{0.2309} & 0.7531 \\
$2^{18}$ & 0.7043 & \textbf{0.3231} & 0.6655 & 1.316 \\
\addlinespace[2pt]
\midrule
\multicolumn{5}{l}{\textbf{$n=10^{8}$}} \\
\addlinespace[1pt]
$2^{7}$ & \textbf{0.9558} & 1.446 & 2.833 & 73.34 \\
$2^{8}$ & \textbf{0.9609} & 1.444 & 2.849 & 73.52 \\
$2^{9}$ & \textbf{0.9677} & 2.561 & 4.088 & 37.00 \\
$2^{10}$ & \textbf{0.9710} & 2.579 & 4.090 & 18.76 \\
$2^{11}$ & \textbf{0.9819} & 2.584 & 4.091 & 9.651 \\
$2^{12}$ & \textbf{0.9934} & 2.579 & 4.089 & 5.085 \\
$2^{13}$ & \textbf{1.008} & 2.565 & 4.089 & 3.135 \\
$2^{14}$ & \textbf{1.012} & 2.597 & 4.084 & 3.517 \\
$2^{15}$ & \textbf{1.187} & 2.610 & 4.078 & 3.982 \\
$2^{16}$ & \textbf{1.714} & 2.606 & 4.078 & 4.202 \\
$2^{17}$ & \textbf{2.975} & 3.414 & 3.953 & 5.798 \\
$2^{18}$ & 5.528 & \textbf{3.421} & 3.994 & 8.642 \\
\addlinespace[2pt]
\midrule
\multicolumn{5}{l}{\textbf{$n=10^{9}$}} \\
\addlinespace[1pt]
$2^{7}$ & \textbf{9.695} & 13.49 & 27.77 & 734.5 \\
$2^{8}$ & \textbf{9.724} & 13.56 & 27.90 & 734.5 \\
$2^{9}$ & \textbf{9.753} & 24.72 & 30.13 & 369.7 \\
$2^{10}$ & \textbf{9.824} & 24.79 & 30.13 & 187.4 \\
$2^{11}$ & \textbf{9.897} & 24.76 & 30.12 & 96.07 \\
$2^{12}$ & \textbf{9.970} & 24.78 & 30.12 & 50.34 \\
$2^{13}$ & \textbf{10.15} & 24.81 & 30.09 & 30.97 \\
$2^{14}$ & \textbf{10.20} & 24.92 & 30.02 & 34.32 \\
$2^{15}$ & \textbf{11.88} & 25.00 & 29.93 & 38.46 \\
$2^{16}$ & \textbf{16.79} & 25.01 & 29.86 & 40.16 \\
$2^{17}$ & \textbf{28.55} & 33.28 & 29.75 & 54.50 \\
$2^{18}$ & 53.14 & 33.30 & \textbf{29.75} & 81.03 \\
\end{longtable}
\endgroup

%% file: tables/table_runtime_four_n_gaussian.tex
\begingroup
\small
\setlength{\LTleft}{\fill}
\setlength{\LTright}{\fill}
\setlength{\tabcolsep}{7pt}
\renewcommand{\arraystretch}{1.05}
\begin{longtable}{lrrrr}
\caption{Execution times of the methods at representative input sizes (normal distribution)}
\label{tab:runtime-four-n-gaussian}\\
\toprule
$R$ & RT-CDF [ms] & CUB [ms] & RefHP [ms] & Kolonias et al. [ms] \\
\midrule
\endfirsthead

\multicolumn{5}{c}{Table \thetable\ (continued)}\\
\toprule
$R$ & RT-CDF [ms] & CUB [ms] & RefHP [ms] & Kolonias et al. [ms] \\
\midrule
\endhead

\midrule
\multicolumn{5}{r}{Continued on the next page}\\
\endfoot

\bottomrule
\endlastfoot

\multicolumn{5}{l}{\textbf{$n=10^{6}$}} \\
\addlinespace[1pt]
$2^{7}$ & \textbf{0.01812} & 0.02305 & 0.04751 & 0.7089 \\
$2^{8}$ & \textbf{0.02004} & 0.02487 & 0.04623 & 0.7402 \\
$2^{9}$ & \textbf{0.02118} & 0.1017 & 0.04826 & 0.3779 \\
$2^{10}$ & \textbf{0.02213} & 0.1016 & 0.04948 & 0.3643 \\
$2^{11}$ & \textbf{0.02478} & 0.1019 & 0.05229 & 0.2674 \\
$2^{12}$ & \textbf{0.03195} & 0.1014 & 0.05645 & 0.1660 \\
$2^{13}$ & \textbf{0.04552} & 0.1029 & 0.05638 & 0.1115 \\
$2^{14}$ & \textbf{0.05626} & 0.1037 & 0.05788 & 0.1081 \\
$2^{15}$ & 0.06472 & 0.1040 & \textbf{0.04520} & 0.1659 \\
$2^{16}$ & 0.08306 & 0.1046 & \textbf{0.04046} & 0.2481 \\
$2^{17}$ & 0.1741 & 0.1124 & \textbf{0.09901} & 0.5141 \\
$2^{18}$ & 0.1824 & 0.1130 & \textbf{0.1002} & 0.9306 \\
\addlinespace[2pt]
\midrule
\multicolumn{5}{l}{\textbf{$n=10^{7}$}} \\
\addlinespace[1pt]
$2^{7}$ & \textbf{0.07161} & 0.1645 & 0.2316 & 6.985 \\
$2^{8}$ & \textbf{0.07487} & 0.1686 & 0.2190 & 7.297 \\
$2^{9}$ & \textbf{0.07943} & 0.2746 & 0.2227 & 3.672 \\
$2^{10}$ & \textbf{0.08325} & 0.2758 & 0.2225 & 3.527 \\
$2^{11}$ & \textbf{0.08825} & 0.2745 & 0.2243 & 2.541 \\
$2^{12}$ & \textbf{0.09440} & 0.2756 & 0.2257 & 1.448 \\
$2^{13}$ & \textbf{0.1105} & 0.2775 & 0.2256 & 0.7834 \\
$2^{14}$ & \textbf{0.1247} & 0.2790 & 0.2266 & 0.4759 \\
$2^{15}$ & \textbf{0.1572} & 0.2778 & 0.2256 & 0.4535 \\
$2^{16}$ & \textbf{0.2160} & 0.2725 & 0.7039 & 0.5239 \\
$2^{17}$ & 0.3968 & \textbf{0.3083} & 0.6642 & 0.8433 \\
$2^{18}$ & 0.6153 & \textbf{0.3098} & 0.6592 & 1.319 \\
\addlinespace[2pt]
\midrule
\multicolumn{5}{l}{\textbf{$n=10^{8}$}} \\
\addlinespace[1pt]
$2^{7}$ & \textbf{0.9556} & 1.435 & 2.953 & 69.75 \\
$2^{8}$ & \textbf{0.9613} & 1.442 & 2.850 & 73.03 \\
$2^{9}$ & \textbf{0.9676} & 2.582 & 4.080 & 36.77 \\
$2^{10}$ & \textbf{0.9727} & 2.585 & 4.078 & 35.31 \\
$2^{11}$ & \textbf{0.9803} & 2.578 & 4.091 & 25.43 \\
$2^{12}$ & \textbf{0.9920} & 2.578 & 4.094 & 14.48 \\
$2^{13}$ & \textbf{1.007} & 2.580 & 4.096 & 7.738 \\
$2^{14}$ & \textbf{1.012} & 2.587 & 4.089 & 4.803 \\
$2^{15}$ & \textbf{1.186} & 2.590 & 4.089 & 4.293 \\
$2^{16}$ & \textbf{1.707} & 2.607 & 4.078 & 4.274 \\
$2^{17}$ & \textbf{2.881} & 3.429 & 3.983 & 5.951 \\
$2^{18}$ & 5.520 & \textbf{3.419} & 3.990 & 8.690 \\
\addlinespace[2pt]
\midrule
\multicolumn{5}{l}{\textbf{$n=10^{9}$}} \\
\addlinespace[1pt]
$2^{7}$ & \textbf{9.704} & 13.44 & 29.43 & 696.6 \\
$2^{8}$ & \textbf{9.731} & 13.50 & 29.51 & 729.7 \\
$2^{9}$ & \textbf{9.752} & 24.77 & 32.19 & 367.0 \\
$2^{10}$ & \textbf{9.815} & 24.76 & 32.19 & 352.5 \\
$2^{11}$ & \textbf{9.891} & 24.73 & 34.35 & 253.7 \\
$2^{12}$ & \textbf{9.963} & 24.79 & 32.24 & 144.1 \\
$2^{13}$ & \textbf{10.16} & 24.79 & 32.23 & 76.82 \\
$2^{14}$ & \textbf{10.21} & 24.82 & 32.21 & 47.42 \\
$2^{15}$ & \textbf{12.14} & 24.87 & 33.18 & 41.40 \\
$2^{16}$ & \textbf{17.04} & 24.97 & 33.08 & 40.55 \\
$2^{17}$ & \textbf{27.62} & 33.99 & 31.83 & 56.38 \\
$2^{18}$ & 53.21 & 33.34 & \textbf{31.82} & 80.57 \\
\end{longtable}
\endgroup

%% file: tables/table_runtime_four_n_all-equal.tex
\begingroup
\small
\setlength{\LTleft}{\fill}
\setlength{\LTright}{\fill}
\setlength{\tabcolsep}{7pt}
\renewcommand{\arraystretch}{1.05}
\begin{longtable}{lrrrr}
\caption{Execution times of the methods at representative input sizes (all-equal input)}
\label{tab:runtime-four-n-all-equal}\\
\toprule
$R$ & RT-CDF [ms] & CUB [ms] & RefHP [ms] & Kolonias et al. [ms] \\
\midrule
\endfirsthead

\multicolumn{5}{c}{Table \thetable\ (continued)}\\
\toprule
$R$ & RT-CDF [ms] & CUB [ms] & RefHP [ms] & Kolonias et al. [ms] \\
\midrule
\endhead

\midrule
\multicolumn{5}{r}{Continued on the next page}\\
\endfoot

\bottomrule
\endlastfoot

\multicolumn{5}{l}{\textbf{$n=10^{6}$}} \\
\addlinespace[1pt]
$2^{7}$ & \textbf{0.01813} & 0.01908 & 0.07474 & 2.699 \\
$2^{8}$ & 0.02030 & \textbf{0.01905} & 0.09596 & 2.692 \\
$2^{9}$ & \textbf{0.02140} & 0.09403 & 0.1167 & 2.699 \\
$2^{10}$ & \textbf{0.02225} & 0.09387 & 0.1280 & 2.694 \\
$2^{11}$ & \textbf{0.02495} & 0.09367 & 0.1781 & 2.702 \\
$2^{12}$ & \textbf{0.03237} & 0.09434 & 0.2372 & 2.703 \\
$2^{13}$ & \textbf{0.04279} & 0.09674 & 0.4516 & 2.723 \\
$2^{14}$ & \textbf{0.05246} & 0.09567 & 0.4508 & 2.754 \\
$2^{15}$ & \textbf{0.05246} & 0.09591 & 0.4508 & 2.823 \\
$2^{16}$ & \textbf{0.05248} & 0.09638 & 0.4524 & 2.909 \\
$2^{17}$ & 0.1234 & \textbf{0.09981} & 0.5112 & 3.179 \\
$2^{18}$ & 0.1010 & \textbf{0.1010} & 0.5125 & 3.604 \\
\addlinespace[2pt]
\midrule
\multicolumn{5}{l}{\textbf{$n=10^{7}$}} \\
\addlinespace[1pt]
$2^{7}$ & \textbf{0.07673} & 0.1432 & 0.1461 & 26.76 \\
$2^{8}$ & \textbf{0.07958} & 0.1414 & 0.2218 & 26.76 \\
$2^{9}$ & \textbf{0.08302} & 0.2445 & 0.3717 & 26.76 \\
$2^{10}$ & \textbf{0.08690} & 0.2460 & 0.3782 & 26.76 \\
$2^{11}$ & \textbf{0.09209} & 0.2455 & 0.6833 & 26.82 \\
$2^{12}$ & \textbf{0.09795} & 0.2471 & 0.8374 & 26.77 \\
$2^{13}$ & \textbf{0.1157} & 0.2469 & 0.8814 & 26.83 \\
$2^{14}$ & \textbf{0.1286} & 0.2468 & 1.236 & 26.86 \\
$2^{15}$ & \textbf{0.1446} & 0.2472 & 1.680 & 26.95 \\
$2^{16}$ & \textbf{0.1865} & 0.2417 & 3.496 & 26.99 \\
$2^{17}$ & 0.3432 & \textbf{0.2586} & 4.819 & 27.29 \\
$2^{18}$ & 0.5385 & \textbf{0.2593} & 4.818 & 27.79 \\
\addlinespace[2pt]
\midrule
\multicolumn{5}{l}{\textbf{$n=10^{8}$}} \\
\addlinespace[1pt]
$2^{7}$ & \textbf{0.9559} & 1.430 & 1.831 & 267.4 \\
$2^{8}$ & \textbf{0.9613} & 1.430 & 1.852 & 267.4 \\
$2^{9}$ & \textbf{0.9688} & 2.562 & 3.042 & 267.4 \\
$2^{10}$ & \textbf{0.9736} & 2.576 & 3.042 & 267.4 \\
$2^{11}$ & \textbf{0.9806} & 2.576 & 4.093 & 266.7 \\
$2^{12}$ & \textbf{0.9937} & 2.507 & 4.450 & 267.4 \\
$2^{13}$ & \textbf{1.018} & 2.509 & 6.839 & 267.4 \\
$2^{14}$ & \textbf{1.029} & 2.509 & 7.056 & 267.9 \\
$2^{15}$ & \textbf{1.200} & 2.509 & 11.72 & 267.7 \\
$2^{16}$ & \textbf{1.688} & 2.506 & 11.66 & 267.9 \\
$2^{17}$ & \textbf{2.767} & 3.387 & 13.26 & 270.2 \\
$2^{18}$ & 5.333 & \textbf{3.404} & 15.74 & 273.1 \\
\addlinespace[2pt]
\midrule
\multicolumn{5}{l}{\textbf{$n=10^{9}$}} \\
\addlinespace[1pt]
$2^{7}$ & \textbf{9.701} & 13.33 & 17.29 & -- \\
$2^{8}$ & \textbf{9.731} & 13.33 & 17.29 & -- \\
$2^{9}$ & \textbf{9.767} & 24.55 & 19.78 & -- \\
$2^{10}$ & \textbf{9.830} & 24.55 & 21.87 & -- \\
$2^{11}$ & \textbf{9.904} & 24.54 & 23.98 & -- \\
$2^{12}$ & \textbf{9.969} & 24.54 & 21.88 & -- \\
$2^{13}$ & \textbf{10.15} & 24.55 & 30.26 & -- \\
$2^{14}$ & \textbf{10.22} & 23.87 & 46.93 & -- \\
$2^{15}$ & \textbf{12.05} & 23.87 & 75.89 & -- \\
$2^{16}$ & \textbf{17.06} & 23.87 & 79.74 & -- \\
$2^{17}$ & \textbf{27.66} & 32.90 & 109.9 & -- \\
$2^{18}$ & 51.75 & \textbf{32.90} & 109.6 & -- \\
\end{longtable}
\endgroup

%% file: figures/fig_rtcdf_kernel_breakdown.tex
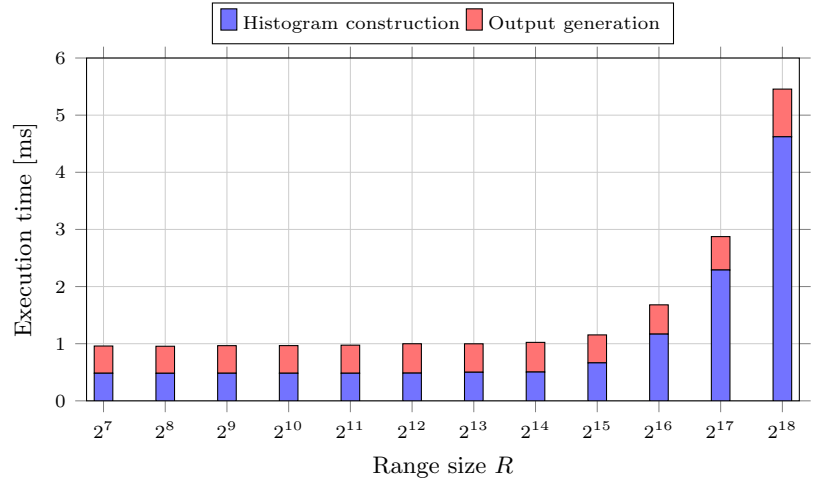
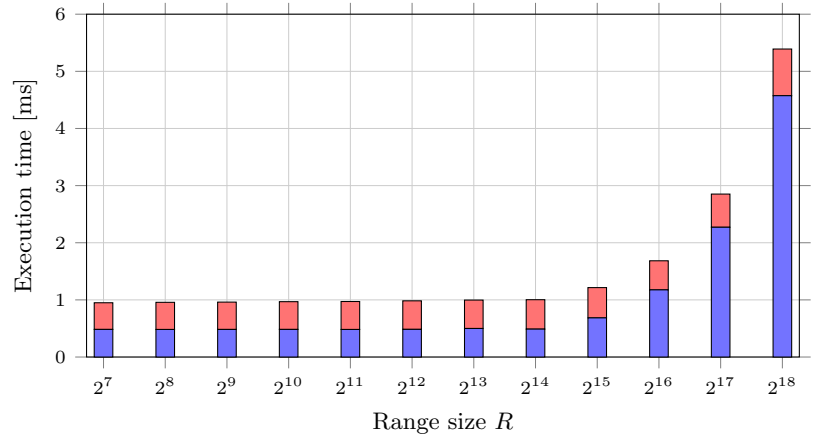
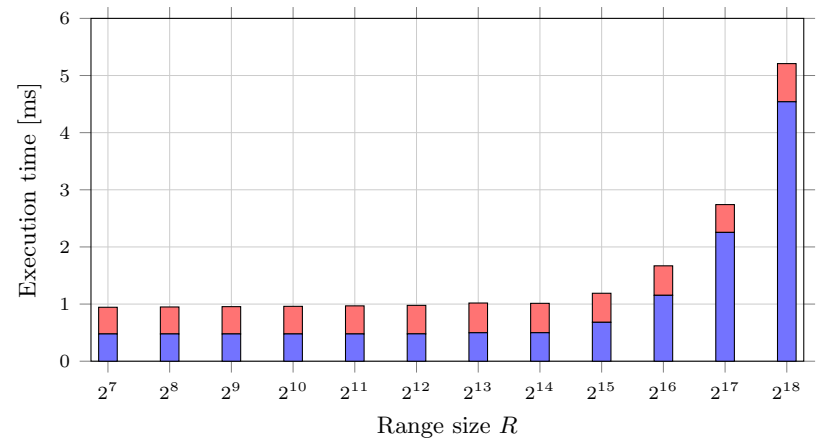
\begin{figure}[H]
\centering
\pgfplotsset{
  rtcdfbreakdownaxis/.style={
    width=0.72\linewidth,
    height=0.400\linewidth,
    ybar stacked,
    bar width=7pt,
    ymin=0, ymax=6.0,
    ytick={0,1,2,3,4,5,6},
    ylabel={Execution time [ms]},
    xlabel={Range size $R$},
    symbolic x coords={7,8,9,10,11,12,13,14,15,16,17,18},
    xtick=data,
    xticklabels={$2^7$,$2^8$,$2^9$,$2^{10}$,$2^{11}$,$2^{12}$,$2^{13}$,$2^{14}$,$2^{15}$,$2^{16}$,$2^{17}$,$2^{18}$},
    xticklabel style={font=\scriptsize},
    yticklabel style={font=\scriptsize},
    label style={font=\small},
    grid=major,
    major grid style={draw=black!20,line width=0.2pt},
    axis line style={draw=black},
    tick align=outside,
    enlarge x limits=0.025,
  }
}

\begin{subfigure}[t]{0.96\linewidth}
\centering
\begin{tikzpicture}
\begin{axis}[rtcdfbreakdownaxis,legend style={at={(0.5,1.04)},anchor=south,legend columns=2,fill=white,draw=black,font=\scriptsize},legend cell align={left}]
\addplot+[fill=blue!55,draw=black] coordinates {(7,0.483328) (8,0.482512) (9,0.483328) (10,0.483328) (11,0.483200) (12,0.485376) (13,0.500736) (14,0.504192) (15,0.664576) (16,1.167840) (17,2.289664) (18,4.621312)};
\addlegendentry{Histogram construction}
\addplot+[fill=red!55,draw=black] coordinates {(7,0.476368) (8,0.473088) (9,0.482496) (10,0.483328) (11,0.491312) (12,0.514048) (13,0.497584) (14,0.519120) (15,0.489472) (16,0.512304) (17,0.584704) (18,0.834560)};
\addlegendentry{Output generation}
\end{axis}
\end{tikzpicture}
\caption{Uniform distribution}
\label{fig:rtcdf-kernel-breakdown-uniform}
\end{subfigure}
\vspace{0.4em}

\begin{subfigure}[t]{0.96\linewidth}
\centering
\begin{tikzpicture}
\begin{axis}[rtcdfbreakdownaxis]
\addplot+[fill=blue!55,draw=black] coordinates {(7,0.483328) (8,0.482304) (9,0.482944) (10,0.483328) (11,0.482304) (12,0.484432) (13,0.499712) (14,0.488448) (15,0.685168) (16,1.175392) (17,2.271232) (18,4.573680)};
\addplot+[fill=red!55,draw=black] coordinates {(7,0.466528) (8,0.475152) (9,0.478304) (10,0.485376) (11,0.489584) (12,0.499200) (13,0.497152) (14,0.514896) (15,0.529920) (16,0.508928) (17,0.580656) (18,0.816640)};
\end{axis}
\end{tikzpicture}
\caption{Normal distribution}
\label{fig:rtcdf-kernel-breakdown-gaussian}
\end{subfigure}
\vspace{0.4em}

\begin{subfigure}[t]{0.96\linewidth}
\centering
\begin{tikzpicture}
\begin{axis}[rtcdfbreakdownaxis]
\addplot+[fill=blue!55,draw=black] coordinates {(7,0.477184) (8,0.477184) (9,0.477184) (10,0.477184) (11,0.477184) (12,0.478656) (13,0.498128) (14,0.498624) (15,0.680960) (16,1.153360) (17,2.254448) (18,4.540416)};
\addplot+[fill=red!55,draw=black] coordinates {(7,0.466944) (8,0.472112) (9,0.477872) (10,0.483328) (11,0.491520) (12,0.498880) (13,0.520304) (14,0.514464) (15,0.508256) (16,0.516032) (17,0.486608) (18,0.667984)};
\end{axis}
\end{tikzpicture}
\caption{All-equal input}
\label{fig:rtcdf-kernel-breakdown-zero}
\end{subfigure}
\caption{Execution times of the principal RT-CDF stages ($n=10^8$)}
\label{fig:rtcdf-kernel-breakdown}
\end{figure}

%% file: sections/05_conclusion.tex
\section{Conclusion}
\label{sec:conclusion}

We focused on counting-sort-based unstable integer sorting, in which output intervals are obtained from a histogram and its prefix sums, and proposed and evaluated RT-CDF, which partitions the value range into tiles that fit in shared memory.
RT-CDF divides the possible integer range into small intervals, constructs a histogram and a local CDF for each tile, and uses the CDF to determine the value corresponding to every output position.
It therefore generates the output array without computing a prefix maximum over the entire output array.

On an NVIDIA GeForce RTX 4090, we compared RT-CDF with CUB \texttt{DeviceRadixSort}, RefHP, and an implementation based on the algorithm of Kolonias et al. for range sizes from $R=2^7$ to $2^{18}$, input sizes from $n=10^6$ to $10^9$, and uniformly distributed, normally distributed, and all-equal inputs.
RT-CDF outperformed the baselines over a broad set of conditions for small to medium ranges.
For every distribution, its maximum speedup over the fastest baseline occurred at $R=2^9$: 3.70 at $n=9\times10^6$ for the uniform distribution, 3.60 at $n=9\times10^6$ for the normal distribution, and 4.39 at $n=10^6$ for the all-equal input.
For large input sizes, RT-CDF also reduced absolute execution time by several to more than ten milliseconds, rather than only improving the speedup ratio.
Its advantage diminished as the range increased.
At $R=2^{17}$, CUB or RefHP was faster for small inputs, but RT-CDF became fastest for large inputs under all three distributions.
At $R=2^{18}$, RT-CDF was not fastest for any evaluated input size or distribution.
Thus, in the evaluated environment, the performance boundary of RT-CDF lies between $R=2^{17}$ and $R=2^{18}$.
RT-CDF performance differed relatively little among the three input distributions and depended more strongly on the range size.

The timing breakdown showed that histogram construction is the primary cause of performance degradation for large ranges.
Because RT-CDF examines the entire input array for every range tile, the number of input references grows as $O(mn)$ with the number of tiles $m$.
At $R=2^{18}$, the larger tile width also increases shared-memory use per thread block, reducing the number of concurrently resident blocks and lowering occupancy.
As a result, histogram construction occupied most of the measured time at $R=2^{18}$.
Output-generation time, by contrast, varied little with the input distribution and did not degrade substantially even when elements were concentrated in a small portion of the range, as for the normal and all-equal inputs.

Reducing histogram-construction cost is therefore the main challenge in extending RT-CDF to larger ranges.
Because the prefix-maximum method outperformed the binary-search method for some combinations of small inputs and large ranges, adaptively selecting the output-generation method according to the input size and range is another possible improvement.
Finally, because the evaluation used only one NVIDIA GeForce RTX 4090, evaluating optimal tile widths and performance characteristics on other GPU architectures remains future work.

%% file: main.bbl
\begin{thebibliography}{10}

\bibitem{arkhipov2017survey}
Dmitri~I Arkhipov, Di~Wu, Keqin Li, and Amelia~C Regan.
\newblock Sorting with {GPUs}: A survey.
\newblock {\em arXiv preprint arXiv:1709.02520}, 2017.

\bibitem{satish2009gpusort}
Nadathur Satish, Mark Harris, and Michael Garland.
\newblock Designing efficient sorting algorithms for manycore {GPUs}.
\newblock In {\em 2009 IEEE International Symposium on Parallel \& Distributed
  Processing}, pages 1--10. IEEE, 2009.

\bibitem{adinets2022onesweep}
Andy Adinets and Duane Merrill.
\newblock {Onesweep}: A faster least significant digit radix sort for {GPUs}.
\newblock {\em arXiv preprint arXiv:2206.01784}, 2022.

\bibitem{nvidia_cub_radixsort}
{NVIDIA Corporation}.
\newblock {CUB} \texttt{DeviceRadixSort} documentation.
\newblock NVIDIA Documentation.

\bibitem{kolonias2011countsort}
Vasileios Kolonias, Artemios~G Voyiatzis, George Goulas, and Efthymios Housos.
\newblock Design and implementation of an efficient integer count sort in
  {CUDA} {GPUs}.
\newblock {\em Concurrency and Computation: Practice and Experience},
  23(18):2365--2381, 2011.

\bibitem{eisenstat2007sumcrcw}
Stanley~C Eisenstat.
\newblock $o(\log^* n)$ algorithms on a {Sum-CRCW PRAM}.
\newblock {\em Computing}, 79(1):93--97, 2007.

\bibitem{kozakai2021hp}
Seiya Kozakai, Noriyuki Fujimoto, and Koichi Wada.
\newblock Efficient {GPU}-implementation for integer sorting based on histogram
  and prefix-sums.
\newblock In {\em Proceedings of the 50th International Conference on Parallel
  Processing}, pages 1--11, 2021.

\bibitem{takase2024hpcasia}
Kaito Takase, Takumi Hagihara, Noriyuki Fujimoto, and Koichi Wada.
\newblock Efficient {GPU}-implementation of a sorting algorithm based on
  histogram computation and prefix-sums.
\newblock {\em Parallel Processing Letters}, 36(01n02):2650005, 2026.

\bibitem{sun2009countsort}
Weidong Sun and Zongmin Ma.
\newblock Count sort for {GPU} computing.
\newblock In {\em 2009 15th International Conference on Parallel and
  Distributed Systems}, pages 919--924. IEEE, 2009.

\bibitem{faujdar2016countsort}
Neetu Faujdar and SatyaPrakash Ghrera.
\newblock Performance evaluation of parallel count sort using {GPU} computing
  with {CUDA}.
\newblock {\em Indian Journal of Science and Technology}, 9(15):1--12, 2016.

\bibitem{sakharnykh2015histogram}
Nikolay Sakharnykh.
\newblock {GPU} pro tip: Fast histograms using shared atomics on {Maxwell}.
\newblock NVIDIA Technical Blog, 2015.

\bibitem{nugteren2011histogram}
Cedric Nugteren, Gert-Jan van~den Braak, Henk Corporaal, and Bart Mesman.
\newblock High performance predictable histogramming on {GPUs}: exploring and
  evaluating algorithm trade-offs.
\newblock In {\em Proceedings of the Fourth Workshop on General Purpose
  Processing on Graphics Processing Units}, pages 1--8, 2011.

\bibitem{henriksen2020generalizedhistogram}
Troels Henriksen, Sune Hellfritzsch, Ponnuswamy Sadayappan, and Cosmin Oancea.
\newblock Compiling generalized histograms for {GPU}.
\newblock In {\em SC20: International Conference for High Performance
  Computing, Networking, Storage and Analysis}, pages 1--14. IEEE, 2020.

\end{thebibliography}
